\documentclass[11pt]{article}

\usepackage[letterpaper,margin=1in]{geometry}
\usepackage{setspace}
\usepackage{amsmath,amssymb,amsfonts,amsthm,mathtools}
\usepackage{newtxtext,newtxmath}
\usepackage{graphicx,tabularx,environ,array,booktabs,float}
\usepackage{enumitem}
\usepackage[numbers,square,sort&compress]{natbib}
\usepackage[hidelinks]{hyperref}
\usepackage{microtype}
\newtheorem{theorem}{Theorem}
\newtheorem{lemma}[theorem]{Lemma}

\newtheorem{corollary}[theorem]{Corollary}
\newtheorem{claim}[theorem]{Claim}
\newtheorem{observation}[theorem]{Observation}
\theoremstyle{definition}

\newtheorem{remark}[theorem]{Remark}

\newenvironment{keywords}{\par\medskip\noindent\textbf{Keywords: }\ignorespaces}{\par\medskip}

\newcommand{\w}{\omega}
\newcommand{\domMLsum}{\textsc{dominant-multilayer-aggregated}}
\newcommand{\domunsum}{\textsc{dominant-uncertain-aggregated}}

\newcommand{\HAT}{\textsc{hat}}

\makeatletter
\newcommand{\problemtitle}[1]{\gdef\@problemtitle{#1}}
\newcommand{\probleminput}[1]{\gdef\@probleminput{#1}}
\newcommand{\problemquestion}[1]{\gdef\@problemquestion{#1}}
\NewEnviron{problem}{
  \problemtitle{}\probleminput{}\problemquestion{}
  \BODY
  \par\addvspace{.5\baselineskip}
  \noindent
  \begin{tabularx}{\textwidth}{@{\hspace{\parindent}} l X c}
    \multicolumn{2}{@{\hspace{\parindent}}l}{\@problemtitle} \\
    \textbf{Input:} & \@probleminput \\
    \textbf{Question:} & \@problemquestion
  \end{tabularx}
  \par\addvspace{.5\baselineskip}
}

\newcommand{\ONEpopunsum}{\textsc{ha-pop-uncertain-aggregated}}
\newcommand{\ONEpopMLsum}{\textsc{ha-pop-multilayer-aggregated}}
\newcommand{\ONETpopMLsum}{\textsc{hat-pop-multilayer-aggregated}}
\newcommand{\ha}{\textsc{ha}}
\newcommand{\quota}{\textbf{q}}

\newcommand{\vote}{\textit{vote}}
\newcommand{\Dvote}{\Delta}

\newcommand{\domrob}{\textsc{dominant-robust}}

\newcommand{\domML}{\textsc{dominant-multilayer}}
\newcommand{\domun}{\textsc{dominant-uncertain}}

\newcommand{\poprob}{\textsc{popular-robust}}
\newcommand{\ONEpoprob}{\textsc{ha-popular-robust}}

\newcommand{\popMLsum}{\textsc{popular-Multilayer-aggregated}}
\newcommand{\popunsum}{\textsc{popular-uncertain-aggregated}}

\makeatother
\newcommand{\popML}{\textsc{popular-multilayer}}
\newcommand{\popf}{\textsc{restricted-(1,1)-forced-pm}}
\newcommand{\ONEpopML}{\textsc{ha-popular-multilayer}}
\newcommand{\ONEpopun}{\textsc{ha-popular-uncertain}}

\newcommand{\popties}{\textsc{popular-matching-with-ties}}

\newcommand{\pop}{\textsc{pm}}
\newcommand{\popun}{\textsc{popular-uncertain}}
\newenvironment{claimproof}{\par\noindent\underline{Proof:}}{\leavevmode\unskip\penalty9999 \hbox{}\nobreak\hfill\quad\hbox{$\blacksquare$}}

\title{Popular Matchings under Preference Variation and an Algorithm for Popular Common Bases with Integral Comparison Margins}
\author{Gergely Csáji}
\date{}

\begin{document}

\maketitle
\begin{abstract}
Preference information in matching markets may be incomplete, criterion-dependent, or noisy. We study popular and dominant matchings under four forms of preference variation: independent uncertainty, multilayer profiles, bounded swap perturbations, and aggregation across profiles. In one-sided markets, we show that all four models admit polynomial-time algorithms for finding a matching that is popular in every relevant realization, or popular with respect to the aggregate comparison in the aggregation model. The aggregate result follows from our main optimization contribution: a pseudo-polynomial extension of the primal--dual level algorithm for popular common bases from partial-order preferences to bounded integral skew-symmetric comparison margins.
We further extend our polynomial-time algorithms for one-sided markets with ties.

In two-sided markets, the existence problem for popular matchings is NP-hard in all four models, whereas dominant matchings remain tractable under uncertainty and bounded swap perturbations. 

\end{abstract}


\section{Introduction}

Matching models are used to allocate indivisible resources and to form bilateral relationships in applications such as school choice, residency assignment and labor markets. Their outcomes depend on ordinal preference information that is often incomplete, criterion-dependent, or noisy. A market designer may have to commit to an assignment before the relevant preferences are fully resolved, or may wish to choose one outcome that performs well across several scenarios. This paper studies how the representation of such preference variation affects the existence and computation of collectively acceptable matchings.

Our collective criterion is \emph{popularity}, introduced by G\"ardenfors~\cite{gardenfors1975match}. Given two matchings, each agent votes for the matching that gives that agent the better partner, and a matching is popular if it never loses a pairwise majority comparison. Popularity is defined in both two-sided matching markets and one-sided house-allocation markets. In two-sided markets every stable matching is popular, while popular matchings can be larger than stable ones~\cite{am2018good}. In one-sided markets a popular matching need not exist, but existence under one preference profile is decidable in polynomial time~\cite{abraham2007popular}. We also study \emph{dominant matchings}, which are popular matchings that strictly defeat every larger matching~\cite{kavitha2014size}.

We consider four ways in which preferences may vary. In the \emph{uncertainty model}, each agent has an independently specified set of possible preference lists, and one seeks a matching that is popular for every resulting realization. This compact representation may encode exponentially many complete profiles. In the \emph{multilayer model}, a finite list of complete profiles is given explicitly, and the matching must be popular in every layer. In the \emph{robust model}, each reported list may be changed by a bounded number of adjacent swaps, and the outcome must remain popular under all allowed perturbations. Finally, in the \emph{aggregate model}, votes are summed across the listed profiles before two matchings are compared. The first three criteria are universal and therefore conservative; aggregation allows a sufficiently strong advantage in one scenario to compensate for a loss in another.

The aggregate model leads to the main optimization contribution of the paper. Kavitha et al.~\cite{yu2023arborescence} developed a primal--dual level algorithm for a popular common base in the intersection of a partition matroid and an arbitrary matroid when each partition class has a partial-order preference. Aggregating several profiles produces integral pairwise margins that may have different magnitudes and may be cyclic, so the partial-order algorithm does not apply directly. We formulate the \emph{integral-margin popular common base problem} and extend the dual certificates and level algorithm to bounded integral skew-symmetric comparison margins. The resulting algorithm is pseudo-polynomial in the margin bound. It is polynomial for the aggregate matching instances considered here because the profiles are explicitly listed and the maximum margin is at most their number.

\subsection{Contributions and main results}

Our results identify a sharp distinction between one-sided and two-sided markets and, within one-sided markets, between universal and aggregate representations of preference variation.

\begin{enumerate}[leftmargin=*,label=\textnormal{(\roman*)}]
\item \textbf{Integral-margin popular common bases.}
We give a self-contained extension of the popular-common-base framework of Kavitha et al.~\cite{yu2023arborescence}. For bounded integral skew-symmetric margins, we prove a chain certificate with bounded multiplicities and adapt the level algorithm. The algorithm is polynomial in the numerical margin bound. Aggregate house allocation, including weak preferences, reduces directly to this general problem; hence a sum-popular matching can be found in polynomial time for any explicitly listed number of profiles.

\item \textbf{One-sided markets under universal variation.}
For strict preferences, certainly popular matchings can be computed in polynomial time in both the uncertainty and multilayer models. In the robust model, we give a polynomial-size reduction to the uncertain model, so the same algorithm applies. We extend the tractability results to weak preferences. The weak-preference algorithm combines simultaneous maximum-cardinality conditions in the first-choice graphs with a directed reachability characterization of houses that can serve as pseudo-second choices across all relevant profiles.

\item \textbf{Two-sided popular and dominant matchings.}
For popular matchings, existence is NP-complete in each of the uncertainty, multilayer, robust, and aggregate models, even under the restrictions stated in Table~\ref{tab:Results}. Verification remains polynomial-time solvable. For dominant matchings, uncertainty and robustness are polynomial-time solvable through their connection to stable matchings, whereas multilayer and aggregate existence are NP-complete.

\item \textbf{Robust stability.}
For completeness, we also resolve the independently bounded robust-stability question posed by Chen, Skowron, and Sorge~\cite{chen2021matchings}: a stable matching that remains stable when every agent may perform up to the prescribed number of adjacent swaps can be found in polynomial time.
\end{enumerate}

Table~\ref{tab:Results} summarizes the complexity classification.

\begin{table}[H]
\centering
\small
\setlength{\tabcolsep}{4pt}
\begin{tabularx}{\textwidth}{@{}lXXX@{}}
\toprule
Model & One-sided popular & Two-sided popular & Dominant \\
\midrule
Uncertainty
& P$^\ast$ (Thm.~\ref{thm:onesided-popun-P})
& NP-c for $k\ge2$ (Thm.~\ref{thm:popML})
& P (Thm.~\ref{thm:domun}) \\
Multilayer
& P$^\ast$ (Thm.~\ref{thm:onesided-popun-P})
& NP-c for $k\ge2$ (Thm.~\ref{thm:popML})
& NP-c for $k\ge2$ (Thm.~\ref{thm:domML}) \\
Robust
& P (Cor.~\ref{cor:robust-red}; Thm.~\ref{thm:onesided-popun-P})
& NP-c for $k\ge1$ (Thm.~\ref{thm:robustpop})
& P (Cor.~\ref{cor:robust-red};  Thm.~\ref{thm:domun}) \\
Aggregate
& P$^\ast$ for all $k$ (Thm.~\ref{thm:ONEsidedsum-general})
& NP-c for $k\ge2$ (Thm.~\ref{thm:popMLsum})
& NP-c for $k\ge2$ (Thm.~\ref{thm:popMLsum}) \\
\bottomrule
\end{tabularx}
\caption{Summary of the complexity results. P denotes polynomial-time solvability and NP-c denotes NP-completeness. A star means that the one-sided result also holds with ties. Here $k$ denotes the relevant uncertainty parameter: the number of profiles or possible lists, or the adjacent-swap radius.}
\label{tab:Results}
\end{table}

\subsection{Related work}

Popularity was introduced in a paper of Gardenfors back in 1975 \cite{gardenfors1975match} for the two-sided matching markets. He showed that every stable matching is popular (in fact they are minimum size popular matchings), hence such matchings always exist and can be found efficiently.  
Abraham et al. \cite{abraham2007popular} introduced popularity for the one-sided preference model, also referred to as the House Allocation (\textsc{ha}) model. They provided an algorithm for finding a maximum size popular matching in the model and also a characterization of popular matchings. 

Huang and Kavitha \cite{huang2013popular} have shown that maximum-size popular matchings can be found in polynomial time can be found in two-sided markets too. Furthermore, Cseh and Kavitha \cite{cseh2018popular} have shown that a nonempty subset of maximum size popular matchings, called dominant matchings of an instance $I$ are in correspondence with the stable matchings of an other instance $I'$. 

In the roommates model, where there is only a single class of agents and the underlying graph is not necessarily bipartite, finding a popular matching is NP-hard (Faenza et al. \cite{faenza2019popular}, Gupta et al. \cite{gupta2021popular}). The problem also becomes NP-hard for another generalization, where the market is still two-sided, but ties are allowed in the preference lists of the agents (Cseh \cite{cseh2017onesidedties}).  We therefore study ties only in the one-sided model, where the single-profile problem remains tractable.

Models, where the preferences of agents can vary or may be uncertain have been studied in many papers in recent literature. 
Uncertain preferences have been studied by Aziz et al. (\cite{aziz2019paretoUn,aziz2022stable,aziz2020stable}). In \cite{aziz2022stable} they studied both a Lottery model, where there was a probability distribution for each different agent over some possible preferences and a Joint Probability model, where the probability distribution was over only certain preference profiles. In both frameworks, they studied when it is possible to find a matching that is stable with probability one (that is, stable in all possible preferences).  In \cite{aziz2019paretoUn}, they studied similar questions for one-sided preferences and Pareto optimality as a condition instead of stability. Our work is very similar to this, as we consider many of the same problems for one-sided and two-sided popularity as they considered for Pareto-optimality and stability.

Miyazaki and Okamoto \cite{miyazaki2017jointly} studied the jointly stable matching problem, where the goal is to find stable matchings across multiple instances with the same agents and contracts, but different preferences, which is equivalent to finding a matching that is stable with probability 1, if each possible preference profile has some positive probability. They strengthened the result of Aziz et al. \cite{aziz2022stable} and showed that it is NP-hard even for two preference profiles. 

Multilayer preferences in bipartite matching markets have been studied in Chen et al. \cite{chen2018stableML}. They defined $\alpha$-global and $\alpha$-pair stability. The former means we aim to find a matching that is stable in at least $\alpha$ layers and the latter means that each edge can block in at most $\alpha$ layers. The bipartite graph of the instances in each layer is the same, hence if $\alpha$ is the number of layers, then $\alpha$-global stability becomes equivalent to deciding whether there is a stable matching across multiple preference profiles. 
They proved that deciding whether an $\alpha$-global or $\alpha$-pair stable matching exists is NP-hard even with strict preferences. They also showed that for $\alpha$-global stability, the problem becomes polynomial solvable, if one side of the agents have the same preferences in all layers. Bentert et al. \cite{bentert2023MLapproval} studied some similar and new notions in Multilayer preferences, where in each layer the preferences of the agents are approval based, i.e. their preference list is a single tie over some acceptable partners. 

Chen, Skowron, and Sorge \cite{chen2021matchings} studied a notion of $d$-globally-robust-stability. This means that the agents may change the position of another agent in their preference list, by moving them up or down such that the number of such swaps in the preference profile is at most $d$. They have shown that finding $d$-globally-robust stable matching can be done in polynomial time. They posed as an open question, whether finding a $d$-robust stable matching, where each agent is allowed to make $d$ swaps in their preference lists independently is solvable in polynomial time. We answer this question positively. 

 The foundation for our aggregate algorithm is the popular-common-base framework of Kavitha et al.~\cite{yu2023arborescence}, which generalizes popular arborescences and several matching problems~\cite{kavitha2025populararb,kavitha2022popularassign}. Our contribution is to permit bounded integral skew-symmetric margins rather than unit comparisons induced by partial orders.

Finally, independently of our work, Bullinger et al.~\cite{bullinger2024robust} also studied similar problems for popular matchings with two-sided preferences. They considered the problem of finding a matching that is popular in two different instances, which only differ in the preferences of the agents. Their model coincides with our multilayer model with two layers. They provided an algorithm for the case, where the two preference profiles only differ at a single agent and showed NP-hardness in general, even if only four agents have different preferences. The advantage of our reduction is that it only uses ties on one side.

\subsection{Paper Structure}
We start in Section~\ref{sec:prelim} with the sufficient background and preliminaries of our studied questions.
Section~\ref{sec:onesided} considers the one-sided preference model. Section~\ref{sec:ones-unc-multi} provides the polynomial algorithms for the uncertain and multilayer model. Section~\ref{subsec:one-sided-ties} extends these results to weak preferences. Section~\ref{subsec:integral-popular-common-base} introduces the bounded skew-symmetric comparison generalization of the popular common-base algorithm of Kavitha et al.~\cite{kavitha2025populararb} that we use to solve the aggregated model in Section~\ref{subsec:aggr}.
Then, we consider the two-sided preference model in Section~\ref{sec:twosided}. We show in Section~\ref{subsec:verif} that verification is efficiently solvable. For dominant matchings. Section~\ref{subsec:dom} shows that we can find certainly dominant matchings in the uncertain and robust models efficiently, when they exist, while the existence question is NP-hard for the multilayer and aggregated model. Section~\ref{subsec:pop} shows that for popular matchings, all four models lead to NP-hard questions. We conclude and summarize in Section~\ref{sec:concl}.

\section{Preliminaries}\label{sec:prelim}

In this section we describe the basic notions of two-sided and one-sided matching markets and our studied problems.
\subsection{Two-sided markets}
In \textit{two-sided markets}, we are given a bipartite graph $G=(A,B;E)$, where the vertices represent the agents and the edges represent the possible contracts or acceptability relations. Furthermore, for each agent $u\in A\cup B$, we are given a strict linear order (called a \textit{preference list}) $\succ_u$ over $V(u)$, where $V(u)$ denotes the neighbors of $u$ in $G$, i.e., the acceptable partners of $u$. We call the set of these preferences the \textit{preference profile} of the instance and denote it by $L$, so $L=\{ \succ_u\mid u\in A\cup B\}$.  

We say that $M\subseteq E$ is a \textit{matching}, if $|M(u)|\le 1$ for all $u\in A\cup B$, where $M(u) \subseteq V(u)$ denotes the partner(s) of $u$ in $M$. If $u$ does not have any partners in $M$, then $M(u)=\emptyset$, which is considered strictly worse than any acceptable partner $v\in V(u)$. Depending on the context, we may also refer by $M(u)$ to the edge adjacent to $u$ in $M$.

\textbf{Stability.}
In two-sided matching markets, the most well-known optimality concept is stability. Let $M$ be a matching. We say that an edge $(a,b)$ \textit{blocks} $M$, if $b\succ_a M(a)$ and $a\succ_b M(b)$ both hold. We say that $M$ is \textit{stable}, if no edge blocks $M$. A well-known result is that a stable matching always exists and can be found in linear time~\cite{GS62}. Furthermore, a so-called \textit{$A$-optimal stable matching} also exists. In this $A$-optimal stable matching, each agent $a\in A$ obtains his best possible partner in any stable matching simultaneously.

\textbf{Popularity.} 
To define the popularity of a matching, take two arbitrary matchings $M$ and $N$. Then, we let each agent $u\in A\cup B$ cast a vote $\vote^L_u (M,N)$ as follows. If $M(u)\succ_u N(u)$, then $\vote^L_u(M,N)=+1$, if $M(u)=N(u)$, then $\vote^L_u(M,N)=0$ and if $N(u)\succ_u M(u)$, then $\vote^L_u (M,N)=-1$. We sum up these votes to obtain $\Delta^L (M,N) = \sum_{u\in A\cup B}\vote^L_u(M,N)$. 

We say that a matching $M$ is \textit{popular} (with the preference profile $L$), if for any other matching $N$ it holds that $\Dvote^L (M,N)\ge 0$. In other words, a matching $M$ is popular if it does not lose any head-to-head comparison with any other matching. Otherwise, if $\Dvote^L (M,N)<0$, we say that $N$ \textit{dominates} $M$.

\textbf{Dominant matchings.}
We say that a matching $M$ is \textit{dominant}, if (i) $M$ is popular and (ii) for any other matching $N$ with $|N|>|M|$ it holds that $\Dvote^L (M,N)>0$, or equivalently $\Dvote^L (N,M)<0$. Dominant matchings are a nonempty subset of maximum-size popular matchings and can be found in polynomial time. The motivation is that, between two matchings receiving equally many favorable votes, the one that matches more agents may still be preferable; dominance incorporates this cardinality consideration.

We will denote the problem of deciding whether an instance $I$ of a two-sided matching market admits the popular matching by \pop. 


Even though \pop\ always admits a solution that can be found in polynomial time~\cite{gardenfors1975match}, the problem becomes NP-hard if we want to decide whether there is a popular matching that contains a given edge $e$ and covers a given agent $z$. For our main hardness result, we will use the following restricted, but still NP-hard variant of this problem. 

\begin{problem}
  \problemtitle{\popf}
  \probleminput{An instance $I$ of \pop; an edge $e=(x,y)\in E$ with $x\in A,y\in B$, such that $V(x)=\{ a,y \}$, $V(a)=\{ x\}$ and $x,y$ are first in each other's preference list; and a vertex $z\in B \setminus \{ a,y\}$. }
  \problemquestion{Is there a popular matching $M$ in $I$ that contains $(x,y)$ and covers $z$?}
\end{problem}

The hardness of this problem was shown in \cite{forcedpop}. They did not state all these restrictions in their theorem, but their construction does satisfy them. 

\begin{theorem}[Faenza, Powers and Zhang \cite{forcedpop}]
    \popf\ is NP-complete.
\end{theorem}

We proceed by defining the problems related to two-sided markets that we study. For the sake of uniformity, use the parameter $k$ for all four models to measure the level of uncertainty. Note however, that this measure differs in the models, hence $k$ refers to different parameters in them.  

We study four different models. In the first, Multilayer Model, we assume that we are given multiple instances of \pop, $I_1,I_2,\dots,I_k$, with the same underlying bipartite graph, but with different preference profiles $L_1,L_2,\dots, L_k$. Our goal is to decide whether there is a matching that is popular  in all preference profiles. We call such a matching a \textit{certainly popular matching}. This model corresponds to the Multilayer preference model in Chen et al. \cite{chen2018stableML}, defined for stability instead of popularity with layers $L_1,L_2,\dots,L_k$, hence we call it \popML. It also corresponds to finding a certainly stable matching in the joint probability model in Aziz et al. \cite{aziz2022stable}. 

\begin{problem}
  \problemtitle{\popML}
  \probleminput{A number $k$, a bipartite graph $G=(A,B;E)$, and preference profiles $L_i=\{ \succ_u^i\mid u\in A\cup B\}$ for $i\in [k]$.}
  \problemquestion{Is there a matching $M$ that is certainly popular?}
\end{problem}

The problem \domML\ is defined analogously, but instead of a certainly popular matching, we are looking for a certainly dominant matching, that is, a matching which is dominant in all of the preference profiles.

We proceed by defining the Uncertainty Model. Instead of one bipartite graph with multiple preference profiles, now we assume that for each agent $u$, there is a set of possible preference lists $P_u= \{ \succ_u^1,\dots,\succ_u^{k_u}\}$. This can be interpreted as each agent having uncertain preferences: we know some information about them, but not everything. In this model, we aim to find a matching that is popular for every preference profile $L\in \prod_{u\in A\cup B}P_u$, that is, for every choice $\succ_u\in P_u$ of the individual preference of each agent $u$. We call such a matching a \textit{certainly popular matching}. This corresponds to the lottery model in Aziz et al.~\cite{aziz2022stable} for stability. 

\begin{problem}
  \problemtitle{\popun}
  \probleminput{A number $k$ and a bipartite graph $G=(A,B;E)$, with a set of preference lists $P_u$ for each $u\in A\cup B$ with $|P_u|\le k$.}
  \problemquestion{Is there a matching $M$ that is certainly popular?}
\end{problem}

The problem \domun\ is defined analogously, but instead of a certainly popular matching, we are looking for a certainly dominant matching, that is, a matching which is dominant in all possible preference profiles. 

The next model is the Robust Model. Let $k\ge 0$. Here, we have a fixed $\succ_u$ preference list for each $u\in A\cup B$ and we know that agent $u$'s real preference list differs from $\succ_u$ by at most $k$ swaps, where a swap switches the order of two adjacent entries in the preference list. For example, $a_2\succ a_3\succ a_1$ can be obtained from $a_1\succ a_2\succ a_3$ with two consecutive swaps $(a_1,a_2)$ and $(a_1,a_3)$. We say that agent $v$ is \textit{$k$-better ($k$-worse)} than $v'$ for $u$, if $v$ is better (worse) than $v'$ even after $k$ arbitrary swaps in $u$'s preference list, respectively. For two preference lists $\succ_u^1,\succ_u^2$, we let $d^s(\succ_u^1,\succ_u^2)$ be the minimum number of swaps required to reach $\succ_u^2$ from $\succ_u^1$. We remark that this is also known as the Kendall--Tau distance of the two lists. We say that a matching $M$ is $k$-\textit{robust popular}, if $M$ is popular in any preference profile $L=\{ \succ_u^L\mid u\in A\cup B\}$, where $d^s(\succ_u,\succ_u^L)\le k$ for all $u\in A\cup B$.

\begin{problem}
  \problemtitle{\poprob}
  \probleminput{A number $k$ and a bipartite graph $G=(A,B;E)$ with preference lists $\succ_u$, $u\in A\cup B$.}
  \problemquestion{Is there a matching $M$ that is $k$-robust popular?}
\end{problem}

We obtain the problem \domrob\ similarly, where we are looking for a $k$-robust dominant matching instead, that is, a matching which is dominant even if each agent may do at most $k$ swaps in his preference list.
 
Finally, we define a new Aggregation Model, where instead of aiming to find matchings that are popular in multiple instances, the goal is to find a matching that is popular after their votes have been aggregated. We define this for both the uncertainty and the multilayer model as follows. For the multilayer model, let $\Delta^{ML}(M,N)=\sum_{i=1}^k\Delta^{L_i}(M,N)$. We say that a matching $M$ is \textit{sum-popular} if $\Delta^{ML}(M,N)\ge 0$ for every matching $N$. We say that $M$ is \textit{sum-dominant} if it is (i) sum-popular and (ii) for every matching $N$ with $|N|>|M|$, it holds that $\Delta^{ML}(M,N)>0$.
Hence, we obtain the following problems. 
\begin{problem}
  \problemtitle{\popMLsum}
  \probleminput{A number $k$, a bipartite graph $G=(A,B;E)$, and preference profiles $L_i=\{ \succ_u^i\mid u\in A\cup B\} $, $i=1,\dots,k$.}
  \problemquestion{Is there a matching $M$ that is sum-popular?}
\end{problem}
Also, \domMLsum\ is the variant where we are looking for a sum-dominant matching instead.

In the uncertainty model, we first sum each agent's votes over that agent's possible preference lists and then sum over the agents. Hence, we set $\Delta^{un}(M,N) = \sum_{u\in A\cup B}\sum_{\succ_u \in P_u}\vote^{\succ_u}_u(M,N)$. For fairness, we assume here that $|P_u|=k$ for all $u\in A\cup B$, so each agent casts the same number of votes; repeated preference lists are allowed in $P_u$. Sum-popular and sum-dominant matchings are defined analogously using $\Delta^{un}$ instead of $\Delta^{ML}$.

\begin{problem}
  \problemtitle{\popunsum}
  \probleminput{A number $k$, a bipartite graph $G=(A,B;E)$, and a set of preference lists $P_u$ for each $u\in A\cup B$ with $|P_u|=k$.}
  \problemquestion{Is there a matching $M$ that is sum-popular?}
\end{problem}
\domunsum\ can be defined analogously.

Index the lists in every $P_u$ arbitrarily and let $L_i$ contain the $i$-th list of every agent. Then $\Delta^{un}(M,N)=\Delta^{ML}(M,N)$. Conversely, taking $P_u=\{\succ_u^i:i\in[k]\}$ from the profiles $L_1,\ldots,L_k$ gives the same equality. Thus \popMLsum\ and \popunsum, as well as \domMLsum\ and \domunsum, are equivalent problems.

\subsection{One-sided markets}

\textit{One-sided markets}, or (capacitated)-\textit{House Allocation (\ha) markets} as often called in the literature are similar to two-sided markets. Here too, we are given a bipartite graph $G=(A,B;E)$. However, now only the agents $a\in A$ have preference lists $\succ_a$ over their acceptable houses $V(a)\subseteq B$, so $L=\{ \succ_a\mid a\in A\}$. We also allow the agents in $B$, who we call houses from now on to have a positive integer capacity $\quota [b]$. 

In this model, we say that $M\subseteq E$ is a \textit{matching}, if it holds that $|M(a)|\le 1$ for any agents $a\in A$ and $|M(b)|\le \quota [b]$ for any house $b\in B$, where $M(a)$ and $M(b)$ denote the house(s) allocated to $a$ and the agent(s) allocated to $b$ in $M$ respectively. 

Popularity for the \ha\ model is defined analogously, with the only difference being that only the agents in $A$ participate in the voting procedure, so $\Delta^L (M,N)$ is defined as $\sum_{a\in A}\vote^L_a(M,N)$. We say that a matching $M$ is \textit{popular}, if $\Delta^L (M,N)\ge 0$ for any matching $N$. 

In the \ha\ model, popular matchings may fail to exist, but it is possible to find a maximum size popular matching in polynomial time, if one exists~\cite{abraham2007popular}. 

We proceed by defining the analogous variants of \popML, \popun, \poprob, \popunsum\ and \popMLsum. 

We start with \popML. Here similarly, we are given a bipartite graph $G=(A,B;E)$ representing the market and a set of preference profiles $L_1=\{ \succ_a^1\mid a\in A\}, \dots, L_k = \{ \succ_a^k\mid a\in A\}$. The question is whether there exists a matching $M$ that is popular in all of these profiles (or layers). Again, we call such a matching a \textit{certainly popular} matching.

\begin{problem}
  \problemtitle{\ONEpopML}
  \probleminput{A number $k$, a bipartite graph $G=(A,B;E)$, capacities $\quota [b]$ for each house $b\in B$ and one-sided preference profiles $L_i=\{ \succ_a^i\mid a\in A \}$, for $i\in [k]$. }
  \problemquestion{Is there a matching $M$ that is certainly popular?}
\end{problem}

For the uncertainty model, there is a set of possible preference lists $P_a=\{ \succ_a^i\mid i=1,\dots,k_a\}$ for each agent $a$, and the question is whether there exists a matching that is popular for any possible choices of $\succ_a\in P_a$, for $a\in A$. 

\begin{problem}
  \problemtitle{\ONEpopun}
  \probleminput{ A number $k$, a bipartite graph $G=(A,B;E)$, capacities $\quota [b]$ for each house $b\in B$ and a set of preference lists $P_a$ for each $a\in A$ with $|P_a|\le k$.}
  \problemquestion{Is there a matching $M$ that is certainly popular?}
\end{problem}

The notion of $k$-robustness can also be defined similarly. Suppose we are given $\succ_a$ preference lists for each $a\in A$. A matching $M$ is said to be $k$-\textit{robust popular}, if it is popular in any preference profile $L=\{\succ_a^L\mid a\in A\}$, where $d^s(\succ_a,\succ_a^L)\le k$ for all $a\in A$.

\begin{problem}
  \problemtitle{\ONEpoprob}
  \probleminput{A number $k$, a bipartite graph $G=(A,B;E)$ with preference lists $\succ_a$ for $a\in A$ and capacities $\quota [b]$ for $b\in B$.}
  \problemquestion{Is there a matching $M$ that is $k$-robust popular?}
\end{problem}

By defining $\Delta^{ML}$ and $\Delta^{un}$ analogously as in two-sided markets, we define the problems \ONEpopMLsum\ and \ONEpopunsum. For the sake of fairness, we again assume here that $|P_a|=k$ for all $a\in A$, so each agent casts the same number of votes and hence we allow an agent to have the same preference lists multiple times in $P_a$. 

\begin{problem}
  \problemtitle{\ONEpopMLsum}
  \probleminput{A number $k$, a bipartite graph $G=(A,B;E)$, capacities $\quota[b]$ for each $b\in B$, and preference profiles $L_i=\{ \succ_a^i\mid a\in A\} $, $i\in [k]$.}
  \problemquestion{Is there a matching $M$ that is sum-popular?}
\end{problem}


\begin{problem}
  \problemtitle{\ONEpopunsum}
  \probleminput{A number $k$, a bipartite graph $G=(A,B;E)$, capacities $\quota[b]$ for each $b\in B$, and a set of preference lists $P_a$ for each $a\in A$ with $|P_a|=k$.}
  \problemquestion{Is there a matching $M$ that is sum-popular?}
\end{problem}

Index the lists in every $P_a$ arbitrarily and let $L_i$ contain the $i$-th list of every applicant. Then $\Delta^{un}(M,N)=\Delta^{ML}(M,N)$. Conversely, taking $P_a=\{\succ_a^i:i\in[k]\}$ from the profiles $L_1,\ldots,L_k$ gives the same equality. Thus \ONEpopMLsum\ and \ONEpopunsum\ are equivalent problems too.

\paragraph{Weak preferences.}
We also study the extension in which an applicant's preference list is a weak order, denoted by $\succeq_a$, and may therefore contain ties. An applicant votes only when one assignment is strictly preferred to the other and abstains when the two houses are tied. We denote this house-allocation-with-ties model by \HAT, and use the prefix \textsc{hat} for its uncertainty, multilayer, robust and aggregated variants.

\subsection{Matroids and common bases}

A \emph{matroid} is a pair $\mathcal M=(E,\mathcal I)$ consisting of a finite ground set $E$ and a nonempty family $\mathcal I\subseteq 2^E$ of \emph{independent sets} satisfying two axioms: every subset of an independent set is independent, and whenever $X,Y\in\mathcal I$ with $|X|<|Y|$, some $e\in Y\setminus X$ satisfies $X+e\in\mathcal I$. A maximal independent set is a \emph{base}; all bases have the same cardinality, called the rank of $\mathcal M$. More generally, the rank function is
\[
    \rho(S)=\max\{|I|:I\subseteq S,\ I\in\mathcal I\}
    \qquad(S\subseteq E).
\]
For $S\subseteq E$, its span, or closure, is
\[
    \operatorname{span}_{\mathcal M}(S)
      =\{e\in E:\rho(S+e)=\rho(S)\}.
\]
A set $S$ is a \emph{flat} if $\operatorname{span}_{\mathcal M}(S)=S$.

Given a partition $E=\mathbin{\dot\bigcup}_{a\in A}E_a$, the associated \emph{partition matroid} consists of the sets $I\subseteq E$ satisfying $|I\cap E_a|\le1$ for every $a\in A$. A set that is independent in two matroids on the same ground set is a \emph{common independent set}; a set that is a base of both is a \emph{common base}. The weighted matroid-intersection problem asks for a maximum-weight common independent set and is solvable in polynomial time when the matroids are given by independence oracles. Section~\ref{subsec:integral-popular-common-base} uses these notions to state and solve the integral-margin popular common base problem.

\subsection{Reduction of the Robust Model to the Uncertainty Model}
\label{subsec:robust-to-uncertainty}

We show that the problem of deciding whether a $k$-robust popular
matching, or a $k$-robust dominant matching, exists reduces in polynomial
time to the corresponding problem in the uncertainty model. The
reduction applies to both one-sided and two-sided preferences.

Let $I$ be an instance of a robust popular- or dominant-matching problem,
and let $\succ_u$ be the original preference list of agent $u$. We assume
that the distance between two preference lists is the minimum number of
adjacent swaps required to transform one into the other.

For every acceptable partner $v\in V(u)$, let $\succ_u^v$ be the
preference list obtained from $\succ_u$ by moving $v$ downward by
\[
    \min\{k,|\{x\in V(u):v\succ_u x\}|\}
\]
positions, while preserving the relative order of all other agents.
Thus, $\succ_u^v$ is obtained from $\succ_u$ using at most $k$ adjacent
swaps. We define
\[
    P_u=\{\succ_u\}\cup\{\succ_u^v:v\in V(u)\}.
\]
The uncertainty instance $I'$ has the same underlying graph as $I$ and
assigns the set $P_u$ of possible preference lists to each agent $u$.

For example, if
\[
    \succ_u\;=\;v_1\succ v_2\succ v_3\succ v_4
    \qquad\text{and}\qquad k=2,
\]
then
\[
    \succ_u^{v_1}\;=\;v_2\succ v_3\succ v_1\succ v_4.
\]

We prove that a matching $M$ is $k$-robust popular in $I$ if and only if
it is certainly popular in $I'$. Fix two matchings $M$ and $N$ and an
agent $u$. The value of
\(
    \vote_{\succ'_u}(M,N)
\)
depends only on the relative order of $M(u)$ and $N(u)$ in $\succ'_u$.
If $u$ is matched in $M$, a worst preference list for $M$, among all
lists at distance at most $k$ from $\succ_u$, can therefore be obtained
by moving $M(u)$ downward by as many as $k$ positions. Moving any other
agent, or changing the relative order of agents other than $M(u)$ and
$N(u)$, cannot decrease $u$'s vote further. Consequently,
\begin{equation}
\label{eq:robust-local-reduction}
    \min_{\substack{\succ'_u\\
          d(\succ_u,\succ'_u)\le k}}
       \vote_{\succ'_u}(M,N)
    =
    \min_{\widehat{\succ}_u\in P_u}
       \vote_{\widehat{\succ}_u}(M,N).
\end{equation}
Indeed, every list in $P_u$ lies in the radius-$k$ neighborhood of
$\succ_u$, and this set contains the list $\succ_u^{M(u)}$, which attains
the minimum on the left-hand side. If $u$ is unmatched in $M$, its vote
cannot be made worse by changing the order of its acceptable partners,
and the original list $\succ_u\in P_u$ attains the minimum.

Since agents' preference perturbations are independent, summing
\eqref{eq:robust-local-reduction} over all agents gives
\[
    \min_{\substack{\succ'_u:
          d(\succ_u,\succ'_u)\le k\\
          \text{for every agent }u}}
       \vote_{\succ'}(M,N)
   =
    \min_{\substack{\widehat{\succ}_u\in P_u\\
          \text{for every agent }u}}
       \vote_{\widehat{\succ}}(M,N).
\]
It follows that $M$ defeats or ties every matching $N$ under every
radius-$k$ perturbation of the original preferences if and only if $M$
does so under every realization of the uncertainty instance $I'$.
Therefore,
\[
    M\text{ is $k$-robust popular in }I
    \quad\Longleftrightarrow\quad
    M\text{ is certainly popular in }I'.
\]

The same equivalence holds for dominant matchings. In addition to being
popular, dominance requires $M$ to defeat every larger matching. Whether
a matching is larger than $M$ depends only on cardinality and is
unaffected by the preference realization. Hence
\[
    M\text{ is $k$-robust dominant in }I
    \quad\Longleftrightarrow\quad
    M\text{ is certainly dominant in }I'.
\]

The construction is polynomial: for every agent $u$, it creates at most
$|V(u)|+1$ preference lists, each of polynomial length.

\begin{corollary}
\label{cor:robust-red}
A polynomial-time algorithm for \popun, \domun, or \ONEpopun\ yields a
polynomial-time algorithm for \poprob, \domrob, or \ONEpoprob,
respectively.
\end{corollary}

\begin{remark}
The same reduction applies to stability. A matching remains stable under
every collection of preference lists at distance at most $k$ from the
original profile if and only if it is certainly stable in the uncertainty
instance constructed above. Since certainly stable matchings can be found
in polynomial time~\cite{aziz2022stable}, this gives a polynomial-time
algorithm for finding a $k$-robust stable matching and answers the open
question posed by Chen et al.~\cite{chen2021matchings}.
\end{remark}

\section{One-sided preferences}\label{sec:onesided}

In this section we discuss our results for the \ha\ model. Section~\ref{sec:ones-unc-multi} provides the polynomial algorithms for the uncertain and multilayer model. Section~\ref{subsec:one-sided-ties} extends these results to weak preferences. Section~\ref{subsec:integral-popular-common-base} introduces the bounded skew-symmetric comparison generalization of the popular common-base algorithm of Kavitha et al.~\cite{kavitha2025populararb} that we use to solve the aggregated model in Section~\ref{subsec:aggr}.

\subsection{Strict preferences: uncertainty and multilayer models}\label{sec:ones-unc-multi}

First, we describe a characterization of popular matchings in \ha\ instances by Manlove and Sng \cite{manlove2006popular}. Let $(G,\quota, L)$ be an instance of the capacitated house allocation problem with preference profile $L$. We can add a "last-resort" house $\hat{b}$ to $B$ with capacity $|A|$ that is appended to the end of each applicant's preference list. Then, the set of popular matchings in this new instance are in one-to-one correspondence with the original ones: we just say that each unmatched agent $a$ is now at an imaginary worst house, the "last-resort". Hence, from now on we may assume that each instance of \ONEpopML, \ONEpopun\ and \ONEpopMLsum\ contains such a house. 

Very importantly though, in \ONEpopML, \ONEpopun\ and \ONEpopMLsum, the new inputs are such that this house must be always last for every agent in any possible preference list or preference profile in the input. This is without loss of generality, as this house did not exist in the original input.

Next, for each applicant $a\in A$ define the \textit{first house} $f_L(a)$ to be the house that is the most preferred by $a$. For a house $b$, we call an agent $a$ an \emph{admirer of $b$} in $L$, if $b=f_L(a)$. Denote the set of admirers of $b$ by $AD_L(b)$. 
Next, we define the \textit{pseudo-second house} $s_L(a)$ to be either $f_L(a)$, whenever $b=f_L(a)$ has at most as much admirers as its capacity $\quota [b]$, or otherwise as the house most preferred by $a$ among the houses that have less admirers than their capacity. Note that because of the last resort house, $s_L(a)$ exists. 

Make a graph $G_L=(A,B,E_L)$ by adding the edges $(a,f_L(a))$ and $(a,s_L(a))$ for each applicant $a\in A$. In the case when $f_L(a)=s_L(a)$ we only add one of the parallel edges. 
\begin{theorem}[Manlove and Sng\cite{manlove2006popular}]
\label{thm:popperf}
    A matching $M$ is popular if and only if the followng conditions hold:
    \begin{itemize}
        \item Every edge of $M$ is from $E_L$,
        \item $M$ is $A$-perfect, so it matches every $a\in A$ and
        \item every house $b\in B$ that can be saturated with admirers of $b$ only is saturated with admirers of $b$ only.
        
    \end{itemize}
\end{theorem}

Now we consider an instance $I$ of \ONEpopML\ or \ONEpopun.

For a house $b$, let $CA(b)$ denote the set of \textit{certain admirers} of $b$, that is the agents $a\in A$ that consider $b$ first in any possible preference profile/preference list. Let $PA(b)$ denote the set of \textit{possible admirers} of $b$, that is the agents $a\in A$ that consider $b$ best in some possible preference profile/preference list. Finally, let $H_a(b)$ be the set of houses that are better than $b$ in some possible preference list of $a$. All these sets can be computed in polynomial time in both the \ONEpopML\ and the \ONEpopun\ models.

Using theorem \ref{thm:popperf} we get the following result. 

\begin{theorem}
\label{thm:pop-conditionML}
    Let $G=(A,B;E)$ be a bipartite graph, with $\quota [b]$ capacities for each $b\in B$. Let $\mathcal{L}= \{ L_1,L_2,\dots,L_k \}$ be a set of preference profiles for $A$. Then, a matching $M$ is popular in any $L\in \mathcal{L}$ if and only if 
    \begin{itemize}
        \item $M$ is an $A$-perfect matching in the graph $\hat{G}=(A,B,\hat{E})$, where $\hat{E}=\cap_{i=1}^nE_{L_i}$ and
        \item for any $L\in \mathcal{L}$, any house $b$ for which $|AD_L(b)|\ge \quota [b]$ is saturated by agents from $AD_L(b)$.
    \end{itemize}
\end{theorem}

For the uncertainty model, we provide a slightly different characterization that we are going to utilize later.

\begin{theorem}
\label{thm:pop-conditionun}
    Let $G=(A,B;E)$ be a bipartite graph, with $\quota [b]$ capacities for each $b\in B$. Let $P_{a_1},\dots,P_{a_n}$ be sets of preference lists for the agents $a_i\in A$. Then, a matching $M$ is popular in any $L\in \times_{i=1}^nP_{a_i}$ if and only if 
    \begin{itemize}
        \item $M$ is an $A$-perfect matching in the graph $\hat{G}=(A,B,\hat{E})$, where $\hat{E}=\cap_{i=1}^nE_{L_i}$,
        \item any house $b$ that has at least $\quota [b]$ possible admirers is filled with possible admirers,
        \item any house $b$ that has at least $\quota [b]$ certain admirers is filled with certain admirers,
        \item for each agent $a$ and each possible pseudo-second house $b$ of $a$, if $(a,b)\in M$, then each $b'\in H_a(b)$ is filled with certain admirers.
    \end{itemize}
\end{theorem}
\begin{proof}
Let $M$ be a matching that is popular in any possible preference profile. By Theorem \ref{thm:popperf}, the first point clearly must hold. Also, if $b$ that has at least $\quota [b]$ possible admirers, then there is a possible preference profile, where all of those are admirers, so by Theorem \ref{thm:popperf} again, $b$ must be filled with such agents. 

For the third point, assume that $b$ can be filled with certain admirers, but there is an $a\notin CA(b)$ matched to $b$ (by the previous point, we know $b$ must be filled). Consider a preference profile $L$, where $b$ is not the best house of $a$. Then, we can create a matching $N$ from $M$ by assigning $a$ to his best house $b'$, and a certain admirer $a'$ of $b$ not at $b$ in $M$ (by our assumptions there must be such an agent) to $b$. This may require us to reject an agent $a''$ from $b'$ for $N$ to be a matching, but in any case, both $a$ and $a'$ prefers $N$ to $M$ in $L$ and at most one agent is worse off, so $N$ dominates $M$ in $L$, contradiction. 

For the fourth point, suppose that there is an agent $a$ with a possible pseudo-second house $b$, such that $(a,b)\in M$, but there is a house $b'\in H_a(b)$ that is not filled with certain admirers. Clearly, $b'\ne b$. Then, there is a preference list of $a$, say $\succ_a^1$, where $b'\succ_a^1b$. Also, as $b'$ is not filled with applicants from $CA(b')$, there is an agent $a'\ne a$ at $b$ who is not a certain admirer of $b'$ (or $b'$ is not filled, in which case $a$ can just switch to a better house, contradiction). This implies that there is a preference list of $a'$, say $\succ_{a'}^2$, where $b'$ is not his first choice, so there is a house $b''\ne b'$ better than $b'$ for $a'$. Consider a preference profile $L$ with $\succ_a^1,\succ_{a'}^2\in L$. Then, create a matching $N$ from $M$, by assigning $a$ to $b'$, $a'$ to $b''$ and deleting one edge of $N$ adjacent to $b''$ and some agent $a''\ne a'$ if needed (so if $b''$ was filled in $M$). Then in $L$, both $a$ and $a'$ improve from $M$ in $N$ and at most one agent (namely $a''$) obtains a worse situation, hence $N$ dominates $M$ in $L$, contradiction.

In the other direction, suppose that $M$ satisfies the four points. 
In particular, every agent $a\in A$ is matched in $M$. Then, suppose that an agent $a\in A$ is matched to a house $b$ that is not best in some possible preference list of $a$. As $M$ satisfies all points, we get that then $a$ must have degree two in $\hat{E}$ and $b$ is either a first or a pseudo-second house in any preference list of $a$. Suppose that house $b'\in H_a(b)$ is better than $b$ for $a$ in $\succ_a^i$. By the last point, we get that each such house $b'$ is filled with certain admirers. Hence, any agent $a$ can only improve in any possible preference profile by taking the place of an other agent $a'$ who is at his first choice with respect to any of his preference lists and therefore certainly gets worse off. Hence, for any possible preference profile, the number of improving agents can be at most the number of agents getting worse, which implies that $M$ is popular in all of these profiles.
\end{proof}

Using Theorems \ref{thm:pop-conditionML} and \ref{thm:pop-conditionun} we show that \ONEpopML\ and \ONEpopun\ are polynomial time solvable. 

\begin{theorem}
\label{thm:onesided-popun-P}
    \ONEpopML\ and \ONEpopun\ are solvable in  polynomial time. 
\end{theorem}
\begin{proof}
Let $I$ be an instance of either \ONEpopML\ or \ONEpopun.
Let $\mathcal{L}$ denote the set of possible preference profiles, which in the case of \ONEpopun\ is $\mathcal{L}=\times_{a\in A} P_a$, which can have exponentially many profiles in the input size, otherwise only polynomially many.

Let us start with the case of \ONEpopML. 
By Theorem \ref{thm:pop-conditionML}, it is enough to decide whether $\hat{G}=(A,B,\hat{E})$,  $\hat{E}=\cap_{L\in \mathcal{L}}E_{L}$ admits an $A$-perfect matching or not, satisfying that for any $L\in \mathcal{L}$, any house $b$ for which $|AD_L(b)|\ge \quota [b]$ is saturated by agents from $AD_L(b)$.
We can compute $AD_L(b)$ for each $(b,L)\in B\times \mathcal{L}$ in polynomial time, and we can also compute the graph $\hat{G}$ in polynomial time. Now, for each edge $(a,b)\in \hat{E}$, we check whether there is a preference profile $L$, such that $|AD_L(b)|\ge \quota [b]$, but $a\notin AD_L(b)$, and if yes, then we delete the edge $(a,b)$. If there remains an $A$-perfect matching $M$ in the graph, then if clearly satisfies the first condition and also the second one, as otherwise if $(a,b)\in M$ but $a$ is not an admirer of $b$ in $L$ and $|AD_L(b)|\ge \quota [b]$, then the edge $(a,b)$ has been deleted, contradiction. For the other direction, if there is no $A$-perfect matching remaining, then there can be no matching that is popular in every $L\in \mathcal{L}$, because such a matching would be $A$-perfect and none of its edges would have been deleted.

Consider the case of \ONEpopun. We start by computing the sets $CA(b), PA(b)$ for each house $b\in B$ and the sets $H_a(b)$ for each edge $(a,b)\in E$.

First, we show that we can decide if an edge $e$ is included in $\hat{E}$ or not, and hence we can construct the graph $\hat{G}$ in polynomial time. 

By theorem \ref{thm:popperf}, an edge $e=(a,b)$ is included in $\hat{G}$ if and only if $b\in \{ s_L(a),f_L(a) \}$ for each possible preference profile $L\in \mathcal{L}$. If $b$ is first in all of $a$'s possible preference lists, then $(a,b)\in \hat{E}$ clearly. If not, then $(a,b)\in \hat{E}$, if and only if for each such preference list $\succ_a^i$, each house $b'$ that is better than $b$ for $a$ must have at least $\quota [b']$ many admirers in any possible preference profile containing $\succ_a^i$, so $|CA(b')|\ge \quota [b']$, if $b'$ is not the first in $\succ_a^i$ and $|CA(b')|\ge \quota [b']-1$ if $b'$ is the first in $\succ_a^i$. As the sets $CA(b)$ can be computed in polynomial time for all $b\in B$, we can check for a given $a\in A$ and $b\in B$, whether the conditions for $(a,b)\in \hat{E}$ hold for each $\succ_a^i\in P_a$ preference list of $a$ and therefore we can construct $\hat{G}$ in  polynomial time. 

Then, by point two of Theorem \ref{thm:pop-conditionun}, we know that the edges $(a,b)$ such that $|PA(b)|\ge \quota [b]$, but $a\notin PA(b)$ cannot be included in any matching that is certainly popular, so we can safely delete them from $\hat{G}$. Similarly, we can delete the edges $(a,b)$, such that $|CA(b)|\ge \quota [b]$, but $a\notin CA(b)$ by point three of Theorem \ref{thm:pop-conditionun}. 

Finally, for each agent $a\in A$, and each remaining $(a,b)$ edge, if $b$ is a pseudo-second house of $a$ in some preference profile, then we check whether for each $b'\in H_a(b)$, $|CA(b')|\ge \quota [b']$, which can be done in  polynomial time. If not, then $b'$ cannot be filled with certain admirers and hence $a$ cannot be at $b$ in a certainly popular matching by Theorem \ref{thm:pop-conditionun} point four, so we can delete $(a,b)$ from the graph. Otherwise, for each such $b'$, we already deleted each edge $(a',b')$ if $a'\notin CA(b')$, so whenever a matching in the remaining graph fills $b'$, it must fill it with certain admirers. 

Therefore, we conclude that there is a certainly popular matching if and only if the remaining graph has a matching that is (i) $A$-perfect and (ii) fills every $b\in B$ such that $|PA(b)|\ge \quota [b]$ by Theorem \ref{thm:pop-conditionun}. Such a matching, if exists, can be found efficiently with for example the Hungarian-method. 
\end{proof}

A nice corollary of our above proof is the following, which also implies that the verification whether a matching $M$ is certainly popular can be done in polynomial-time for \ONEpopML\ and \ONEpopun.

\begin{theorem}
\label{thm:ONEsidedcheck}
    Deciding if an instance $I$ with $G=(A,B;E)$ of \ONEpopML\ or \ONEpopun\ admits a certainly popular matching can be reduced in polynomial time to deciding if a bipartite graph $\Tilde{G}=(A,B,\Tilde{E})$ admits an $A$-perfect matching that saturates a given set of vertices in $B$. Also, this graph $\Tilde{G}$ satisfies that each $a\in A$ has degree at most 2. 
\end{theorem}

\subsection{Preferences with ties}\label{subsec:one-sided-ties}

We now extend the uncertainty and multilayer results from strict rankings to weak rankings. With ties, both the set of first houses and the set of pseudo-second houses may contain several objects, so the degree-two reduction used above is no longer available. The replacement combines simultaneous optimization in the first-choice graphs with a reachability test for usable pseudo-second houses.

First, we describe a characterization of popular matchings in \HAT\ instances. Let $(G,L)$ be an instance of \HAT.
Without loss of generality, we can assume that all houses have capacity 1. If a house $b$ has capacity $\quota [b]\le |A|$, we can just create $\quota [b]$ many copies $b^1,\dots,b^{\quota[b]}$ and replace $b$ in the preference lists of the agents with a tie $b^1\sim b^2\sim \cdots \sim b^{\quota[b]}$.

Also without loss of generality, we can add a "last-resort" post $\hat{b}_a$ for each agent $a$ that is ranked strictly worst for $a$, and unacceptable for everyone else and corresponds to being unmatched. Then, the set of popular matchings in this new instance is in a one-to-one correspondence with the popular matchings in the original.

Again, the new inputs are such that this house is always set as strictly last for every agent in any possible preference list or preference profile.

Next, for each applicant $a\in A$ define the set of \textit{first houses} $f_L(a)$ to be the houses that are most preferred by $a$. For a house $b$, we call an agent $a$ an \emph{admirer of $b$} in $L$, if $b\in f_L(a)$. Denote the set of admirers of $b$ by $AD_L(b)$. We denote the set of edges between houses and their admirers in the preference profile $L$ by $E_L^f$, and let the corresponding subgraph be $G_L^f=(A,B,E_L^f)$.

Next, we define the set of \textit{pseudo-second houses} $s_L(a)$ to be the set of most preferred houses among the ones that some maximum size matching in $G_L^f=(A,B,E_L^f)$ avoids. For example, if $\succeq_a= b_1\succ b_2\sim b_3\sim b_4\succ b_5$, where $b_1$ and $b_2$ are covered by all maximum size matchings in $G_L^f$, but $b_3,b_4$ and $b_5$ are not, then $s_L(a)=\{ b_3,b_4\}$.
Note that because of the last resort house, $s_L(a)$ is nonempty. 

Make a graph $G_L=(A,B,E_L)$ by adding the edges $(a,b)$ for each applicant $a\in A$ and each $b\in f_L(a)\cup s_L(a)$.

\begin{theorem}[Abraham et al.~\cite{abraham2007popular}]\label{thm:abraham}
    A matching $M$ is popular if and only if
    \begin{itemize}
        \item $M\cap E_L^f$ is a maximum size matching in $G_L^f=(A,B,E_L^f)$ and
        \item for all agents $a\in A$, $M(a)\in f_L(a)\cup s_L(a)$ (in particular, $M(a)\ne \emptyset$).
     \end{itemize}
\end{theorem}

Next, we develop a polynomial algorithm to decide if a certainly popular matching exists.

\begin{lemma}\label{lem:weight-opt}
Let $G=(A,B,E)$ and
let $\mathcal{L}=\{L_1,\dots,L_k\}$ be a set of layers. Define $\w(e)=\lvert\{i\in[k]:e\in E_{L_i}^f\}\rvert$, the number of layers in which $e$ is a first-choice edge. Let $n_i$ be the size of a maximum matching in $G_{L_i}^f$ for $i\in[k]$.

Then, a matching $M$ in $G$ satisfies $\lvert M\cap E_{L_i}^f\rvert=n_i$ for all $i\in[k]$ if and only if $\w(M)=\sum_{i\in[k]}n_i$. Furthermore, for every matching $M$ in $G$, $\w(M)\le \sum_{i\in[k]}n_i$.
\end{lemma}
\begin{proof}
    Let $M$ be an arbitrary matching in $G$. By definition, the total weight of the matching is $\w(M) = \sum_{e \in M} \w(e)$. 

Since $\w(e)$ counts the number of layers $i \in [k]$ where $e$ is a first-choice edge, we can rewrite the sum using an indicator function $\mathbb{I}_{e \in E_{L_i}^f}$. By changing the order of summation (double counting), we evaluate the sum over both the edges and the layers:
\begin{align*}
\w(M) &= \sum_{e \in M} \sum_{i=1}^k \mathbb{I}_{e \in E_{L_i}^f}= \sum_{i=1}^k \sum_{e \in M} \mathbb{I}_{e \in E_{L_i}^f} = \sum_{i=1}^k \vert M \cap E_{L_i}^f \vert.
\end{align*}

Because $M \cap E_{L_i}^f$ constitutes a valid matching within the layer $E_{L_i}^f$, its size is naturally bounded by the size of the maximum matching in that specific layer. Thus, $\vert M \cap E_{L_i}^f \vert \le n_i$ for all $i \in [k]$. 

Summing these individual bounds across all $k$ layers directly yields the second statement of the lemma:
\[
\w(M) = \sum_{i=1}^k \vert M \cap E_{L_i}^f \vert \le \sum_{i=1}^k n_i.
\]

Finally, because the condition $\vert M \cap E_{L_i}^f \vert \le n_i$ is strictly enforced for every layer, the global equality $\w(M) = \sum_{i=1}^k n_i$ can only be satisfied if and only if there is no deficit in any layer---meaning the local equality $\vert M \cap E_{L_i}^f \vert = n_i$ holds for all $i \in [k]$.
\end{proof}

\begin{remark}\label{rem:uncertain-weight}
    In the uncertainty model with $|P_{a_i}|=k$ for every agent (which we can assume by duplicating lists), the $k^n$ possible profiles need not be enumerated. For $e=(a,b)$, its weight in Lemma~\ref{lem:weight-opt} is $\w(e)=k^{n-1}\sum_{i\in[k]}\mathbb{I}_{b\in f_{\succeq_a^i}(a)}$. Here $f_{\succeq_a^i}(a)$ is the set of first houses for $a$ in $\succeq_a^i\in P_a$. 
\end{remark}

\begin{lemma}\label{lem:HA-maxcomputation}
    We can decide in polynomial time, whether a maximum weight matching $M$ in $G$ satisfies $\w (M)=\sum_{L_i\in \mathcal{L}}n_i$ or not.
\end{lemma}
\begin{proof}
    For the multilayer model, where the number of layers is linear in the input size, this is trivial, as we can independently compute the size of the maximum matching $n_i$ in $G_{L_i}^f$ for each layer $L_i \in \mathcal{L}$ and directly verify if $\w(M) = \sum_{L_i \in \mathcal{L}} n_i$.

    We show we can also decide this efficiently for the uncertainty model. By Lemma~\ref{lem:weight-opt}, $\w(M) = \sum_{L_i \in \mathcal{L}} n_i$ if and only if for all possible preference profiles $L \in \mathcal{L}$, the matching $M_L = M \cap E_L^f$ is of maximum size in $G_L^f$. This is equivalent to stating that there is no profile $L \in \mathcal{L}$ that contains an augmenting path with respect to $M_L$ in $G_L^f$. Because $G_L^f$ is bipartite, any such augmenting path must alternate between edges in $E_L^f \setminus M_L$ and edges in $M_L$, starting at an unmatched house in $M_L$ and ending at an unmatched agent in $M_L$.

    We can search for the existence of such a path across all possible profiles by constructing a directed reachability graph. Define the set of potential starting houses $S \subseteq B$: a house $b \in S$ if either $b$ is unmatched in $M$, or $b$ is matched in $M$ to some agent $a = M(b)$ and there exists a possible preference list $\succeq_a^i \in P_a$ where $b \notin f_{\succeq_a^i}(a)$. 

    Construct a directed graph $D = (B, \vec{E})$, where a directed edge $(b, b') \in \vec{E}$ exists if and only if there is an agent $a \in A$ such that $M(a) = b'$ and there exists a preference list $\succeq_a^i \in P_a$ satisfying $\{b, b'\} \subseteq f_{\succeq_a^i}(a)$. This represents an alternating step $b \to a \to b'$, where $b$ is a first choice of $a$, and $b'$ is both the $M$-partner of $a$ and a first choice of $a$.

    Define a set of terminal houses $T \subseteq B$: a house $b \in T$ if there exists an agent $u \in A$ and a preference list $\succeq_u^i \in P_u$ such that $b \in f_{\succeq_u^i}(u)$, and either $u$ is unmatched in $M$, or $M(u) \notin f_{\succeq_u^i}(u)$.

    We claim that there exists a profile $L \in \mathcal{L}$ where $M_L$ is not maximal if and only if there is a directed path in $D$ from some house in $S$ to some house in $T$.

    For the forward direction, suppose such a profile $L$ exists. Then there is an augmenting path $b_0, a_1, b_1, \dots, a_k$ in $G_L^f$. Since $b_0$ is unmatched in $M_L$, $b_0 \in S$. For each step $b_{i-1} \to a_i \to b_i$, $a_i$ is matched to $b_i$ in $M_L$, so $\{b_{i-1}, b_i\} \subseteq f_L(a_i)$, implying $(b_{i-1}, b_i) \in \vec{E}$. Finally, $a_k$ is unmatched in $M_L$, meaning $M(a_k) \notin f_L(a_k)$ (or $M(a_k) = \emptyset$), while $b_{k-1} \in f_L(a_k)$. Thus, $b_{k-1} \in T$, and a path from $S$ to $T$ exists in $D$.

    Conversely, suppose there is a path from $S$ to $T$ in $D$. Choose a \emph{shortest} such path $b_0 \to b_1 \to \dots \to b_{k-1}$, ensuring the houses $b_i$ are distinct. Let $a_i = M(b_i)$ for $1 \le i \le k-1$. Since $M$ is a matching, the intermediate agents $a_1, \dots, a_{k-1}$ are distinct. Let $u_k$ be the agent satisfying the terminal condition for $b_{k-1}$. 
    
    Crucially, we must have $u_k \notin \{a_1, \dots, a_{k-1}\}$. If $u_k = a_j$ for some internal index $j$, then by the terminal condition definition, there exists a list $\succeq_{a_j}^i \in P_{a_j}$ where $M(a_j) = b_j \notin f_{\succeq_{a_j}}^i(a_j)$. However, this immediately implies $b_j \in S$. The subpath $b_j \to b_{j+1} \dots \to b_{k-1}$ would then be a strictly shorter path from $S$ to $T$, contradicting our assumption. 
    
    If $b_0$ is matched, let $a_0 = M(b_0)$. It is permissible for $u_k = a_0$; the terminal preference list for $a_0$ naturally rejects $b_0$, perfectly satisfying the condition for $b_0 \in S$. Because the active agents determining the path steps ($a_1, \dots, a_{k-1}$) and the terminal agent ($u_k$) are mutually distinct, we can independently select their required preference lists from their respective sets $P_a$ (and pick arbitrary lists for all other agents) to form a valid profile $L \in \mathcal{L}$. In this profile $L$, the sequence $b_0, a_1, b_1, \dots, u_k$ forms an exact augmenting path for $M_L$.

    Since the sets $S$, $T$, and the edges of graph $D$ can be constructed in polynomial time by examining the bounded number of possible preference lists for each agent, and reachability can be solved in linear time relative to the size of $D$ via Breadth-First Search, the entire decision process takes polynomial time.
\end{proof}

\begin{lemma}\label{lem:usable-pseudo-second}
    Let $M$ be a matching in $G$ satisfying $\w(M) = \sum_{L \in \mathcal{L}} n_L$. Let $a \in A$ and fix a preference list $\succeq_a^i \in P_a$. Define the restricted domain $\mathcal{L}_i(a) = \{ \succeq_a^i \} \times \prod_{a' \neq a} P_{a'}$. 
    
    Define the set of starting houses $S^{(i)} \subseteq B$: a house $b_0 \in S^{(i)}$ if either $b_0$ is unmatched in $M$, or $b_0$ is matched to some agent $a' = M(b_0)$ and there exists a possible preference list $\succeq_{a'}^{\ell} \in P_{a'}$ valid in $\mathcal{L}_i(a)$ where $b_0 \notin f_{\succeq_{a'}^{\ell}}(a')$. 
    Construct the directed graph $D^{(i)} = (B, \vec{E}^{(i)})$, where a directed edge $(b, b') \in \vec{E}^{(i)}$ exists if and only if there is an agent $a'$ such that $M(a') = b'$ and there exists a preference list $\succeq_{a'}^{\ell} \in P_{a'}$ valid in $\mathcal{L}_i(a)$ satisfying $\{b, b'\} \subseteq f_{\succeq_{a'}^{\ell}}(a')$. 
    
    Let $R^{(i)} \subseteq B$ be the set of houses reachable from $S^{(i)}$ in $D^{(i)}$. Let $TC^*$ be the most preferred tie-class in $\succeq_a^i \setminus f_{\succeq_a^i}(a)$ that intersects $R^{(i)}$. 
    
    If $M^*$ is a certainly popular matching and $M^*(a) = b \notin f_{\succeq_a^i}(a)$, then:
    \begin{enumerate}
        \item $b \in TC^*$, and
        \item the maximum $\w$-weight of a matching in $G \setminus \{b\}$ over the restricted domain $\mathcal{L}_i(a)$ equals $\sum_{L \in \mathcal{L}_i(a)} n_L$.
    \end{enumerate}
    Furthermore, any house $b$ satisfying these two conditions belongs to $\bigcap_{L \in \mathcal{L}_i(a)} s_L(a)$, and these conditions can be evaluated in polynomial time.
\end{lemma}
\begin{proof}
    Because $\w(M) = \sum_{L \in \mathcal{L}} n_L$, by Lemma~\ref{lem:weight-opt}, the restriction $M_L = M \cap E_L^f$ is a maximum size matching in $G_L^f$ for every profile $L \in \mathcal{L}$, which strictly includes all profiles in $\mathcal{L}_i(a)$. 
    
    We first prove that $b' \in R^{(i)}$ if and only if $b'$ is non-essential (avoided by a maximum size matching in $G_L^f$) in at least one profile $L \in \mathcal{L}_i(a)$. 
    
    $(\Rightarrow)$ Suppose $b' \in R^{(i)}$. Let $b_0 \to b_1 \to \dots \to b_k = b'$ be a shortest path in $D^{(i)}$ from $S^{(i)}$ to $b'$. Let $a_j = M(b_j)$ for $1 \le j \le k$. Because it is a shortest path, the houses $b_0, \dots, b_k$ are distinct, and since $M$ is a matching, the intermediate agents $a_1, \dots, a_k$ are also distinct. By the construction of $\vec{E}^{(i)}$, for each directed edge $(b_{j-1}, b_j)$, there exists a valid preference list $\succeq_{a_j}^{\ell_j} \in P_{a_j}$ such that $\{b_{j-1}, b_j\} \subseteq f_{\succeq_{a_j}^{\ell_j}}(a_j)$. Crucially, because the agents $a_1, \dots, a_k$ are mutually distinct (and agent $a$'s list is fixed to $\succeq_a^i$), we can independently assign list $\succeq_{a_j}^{\ell_j}$ to agent $a_j$ (while assigning arbitrary valid lists to the remaining agents) to construct a specific profile $L \in \mathcal{L}_i(a)$. In this profile $L$, the sequence $b_0, a_1, b_1, \dots, a_k, b_k$ forms an exact alternating path with respect to the maximum matching $M_L$, starting from a house $b_0$ that is either completely unmatched or rejected as a first choice ($b_0 \notin f_{\succeq_{a_0}^{\ell_0}}(a_0)$ where $a_0 = M(b_0)$). As $M\cap E_L^f$ is maximum size, this implies that $b'$ is avoided by some maximum size matching in $G_L^f$, making it non-essential in $L$.

    $(\Leftarrow)$ Conversely, suppose there is a profile $L \in \mathcal{L}_i(a)$ where $b'$ is non-essential. Then, in the first-choice graph $G_L^f$, there exists a (possibly single vertex) augmenting or alternating path with respect to the maximum size matching $M_L$ starting from a house $b_0$ that is unmatched in $M_L$. This $b_0$ is either strictly unmatched in $M$, or matched in $M$ but rejected as a first choice in $L$; in either case, $b_0 \in S^{(i)}$. Every step $b_{j-1} \to a_j \to b_j$ along this path implies that $a_j$ is matched to $b_j$ in $M$, and $a_j$ considers both $b_{j-1}$ and $b_j$ as first choices in $L$. This directly satisfies the definition for the directed edge $(b_{j-1}, b_j) \in \vec{E}^{(i)}$ generated by the specific list $\succeq_{a_j}^{\ell_j}$ assigned to $a_j$ in $L$. Consequently, the alternating path in $G_L^f$ maps to a valid directed path from $S^{(i)}$ to $b'$ in $D^{(i)}$, proving $b' \in R^{(i)}$.

    Now, assume $M^*$ is a certainly popular matching and $M^*(a) = b \notin f_{\succeq_a^i}(a)$. By the characterization theorem (Theorem~\ref{thm:abraham}), for every $L \in \mathcal{L}_i(a)$, $M^*(a) \in f_L(a) \cup s_L(a)$. Since $b \notin f_L(a)$, it must be that $b \in s_L(a)$ for all $L \in \mathcal{L}_i(a)$. This means $b$ is non-essential in \emph{every} $L \in \mathcal{L}_i(a)$. 
    
    \begin{itemize}
        \item \textbf{Evaluating Condition 1:} Since $b$ is non-essential in all $L \in \mathcal{L}_i(a)$, it must be that $b \in R^{(i)}$. We must show $b \in TC^*$. Suppose for contradiction that $b \notin TC^*$. Because $TC^*$ is defined as the most preferred tie-class intersecting $R^{(i)}$, it is impossible for $b \succ_a^i TC^*$. If $b \prec_a^i TC^*$, then since $TC^*$ intersects $R^{(i)}$, there exists some house $b^* \in TC^*$ that is non-essential in some specific profile $L^* \in \mathcal{L}_i(a)$. By the definition of pseudo-second houses, the choices $s_{L^*}(a)$ must be at least as preferred as $b^*$. Thus, $s_{L^*}(a) \succeq_a^i b^* \succ_a^i b$, which strictly contradicts the fact established above that $b \in s_{L^*}(a)$. Therefore, $b \in TC^*$.
        
        \item \textbf{Evaluating Condition 2:} Because $M^*$ is certainly popular, for every $L \in \mathcal{L}_i(a)$, the edges $M^* \cap E_L^f$ form a maximum size matching in $G_L^f$ of size $n_L$. Since $b \notin f_{\succeq_a^i}(a)$, the edge $(a,b)$ is not a first-choice edge in any $L \in \mathcal{L}_i(a)$ (i.e., $(a,b) \notin E_L^f$). Removing the edge $(a,b)$ leaves the matching $M^* \setminus \{(a,b)\}$ in the subgraph $G \setminus \{b\}$. Crucially, this remaining matching still contains exactly $n_L$ first-choice edges for every $L \in \mathcal{L}_i(a)$. By Lemma~\ref{lem:weight-opt}, the maximum $\w$-weight of a matching on $G \setminus \{b\}$ evaluated over $\mathcal{L}_i(a)$ must reach the upper bound of exactly $\sum_{L \in \mathcal{L}_i(a)} n_L$. 
    \end{itemize}
    
    Furthermore, if a house $b$ satisfies both conditions, then the fact that the $\w$-weight on $G \setminus \{b\}$ equals $\sum n_L$ guarantees that there is a common matching avoiding $b$ that maximizes first-choice edges across all layers in $\mathcal{L}_i(a)$. Thus, $b$ is non-essential in \emph{every} $L \in \mathcal{L}_i(a)$. Since $b \in TC^*$, and no house $b' \succ_a^i TC^*$ belongs to $R^{(i)}$ (meaning any such $b'$ is essential in every layer), $b$ is guaranteed to be among the most preferred non-essential houses in every layer. Thus, $b \in \bigcap_{L \in \mathcal{L}_i(a)} s_L(a)$.
    
    Constructing $D^{(i)}$ and $S^{(i)}$ via $M$, computing reachability to find $R^{(i)}$ and $TC^*$, and checking the maximum $\w$-weight on $G \setminus \{b\}$ all execute in polynomial time.
\end{proof}

\paragraph{Algorithm \ref{alg:pop-ties}: Finding a Certainly Popular Matching}
\label{alg:pop-ties}
\begin{enumerate}
    \item Define $\w (e) = \vert \{ i\in [k]: e\in E_{L_i}^f\} \vert$ for the Multilayer model, or $\w(e)$ as in Remark~\ref{rem:uncertain-weight} for the Uncertainty model. Compute a maximum $\w$-weight matching $M$ in $G$.
    \item Verify if $\w(M) = \sum_{L_i\in \mathcal{L}}n_i$ using Lemma~\ref{lem:HA-maxcomputation}. If not, STOP and output ``no solution exists''.
    \item Compute the set of valid assignments $x(a)$ for each agent $a \in A$:
    \begin{itemize}
        \item \textbf{Multilayer model:} For each agent $a$, compute $f_L(a)$ and $s_L(a)$ for all layers $L \in \mathcal{L}$. Set $x(a) = \bigcap_{L \in \mathcal{L}} (f_L(a) \cup s_L(a))$.
        \item \textbf{Uncertainty model:} For each possible list $\succeq_a^i \in P_a$, define the restricted domain $\mathcal{L}_i(a) = \{ \succeq_a^i \} \times \prod_{a' \neq a} P_{a'}$. Initialize $x_i(a) = f_{\succeq_a^i}(a)$. 
        
        To find the valid pseudo-second houses, compute the reachable set $R^{(i)}$ in the graph $D^{(i)}$ over $\mathcal{L}_i(a)$, as established in Lemma~\ref{lem:usable-pseudo-second}. Identify $TC^*$, the most preferred tie-class in $\succeq_a^i \setminus f_{\succeq_a^i}(a)$ that intersects $R^{(i)}$. If no such tie-class exists, $x_i(a) = f_{\succeq_a^i}(a)$.
        
        Otherwise, for each candidate house $b \in TC^*$, temporarily remove $b$ from $G$ and compute the maximum $\w$-weight matching on $G \setminus \{b\}$ over the restricted domain $\mathcal{L}_i(a)$. If this maximum weight achieves the upper bound $\sum_{L \in \mathcal{L}_i(a)} n_L$ (verified via Lemma~\ref{lem:HA-maxcomputation}), add $b$ to $x_i(a)$.
        
        Finally, set $x(a) = \bigcap_{\succeq_a^i \in P_a} x_i(a)$.
    \end{itemize}
    \item Construct a restricted subgraph $G' = (A, B, E')$ where $E' = \{ (a,b) \in E \mid b \in x(a) \}$. Find an $A$-perfect matching $M'$ in $G'$ that maximizes the $\w$-weight. If $M'$ is $A$-perfect and $\w(M') = \sum_{L \in \mathcal{L}} n_L$, output $M'$. Otherwise, STOP and output ``no solution exists''. 
\end{enumerate}

\begin{theorem}\label{thm:alg-correctness}
    Algorithm~\ref{alg:pop-ties} finds a certainly popular matching in polynomial time, if one exists.
\end{theorem}
\begin{proof}
    By the characterization of Theorem~\ref{thm:abraham}, a matching $M^*$ is certainly popular if and only if for every possible profile $L \in \mathcal{L}$: (i) $M^* \cap E_L^f$ is a maximum size matching in $G_L^f$, and (ii) $M^*(a) \in f_L(a) \cup s_L(a)$ for all $a \in A$.

    Condition (i) holds simultaneously for all $L \in \mathcal{L}$ if and only if $\w(M^*) = \sum_{L \in \mathcal{L}} n_L$, by Lemma~\ref{lem:weight-opt}. Step 2 of the algorithm correctly uses Lemma~\ref{lem:HA-maxcomputation} to determine if any such matching can exist in the underlying graph. 

    Condition (ii) requires that for each agent $a \in A$, the assigned house $M^*(a)$ belongs to a valid set $x(a) \subseteq \bigcap_{L \in \mathcal{L}} (f_L(a) \cup s_L(a))$. 
    For the Multilayer model, the number of layers is polynomially bounded, allowing direct calculation of $x(a)$ as the exact intersection.

    For the Uncertainty model, the algorithm computes $x(a) = \bigcap_{\succeq_a^i \in P_a} x_i(a)$. The set $x_i(a)$ is initialized with $f_{\succeq_a^i}(a)$ and is then augmented with exactly the houses in $TC^*$ that preserve the maximum $\w$-weight on $G \setminus \{b\}$. By Lemma~\ref{lem:usable-pseudo-second}, this procedure exhaustively identifies every non-first house $b$ that could possibly be assigned to $a$ in a certainly popular matching. If a house drops the maximum $\w$-weight in $\mathcal{L}_i(a)$, it implies that any matching utilizing $b$ for $a$ strictly deprives the remaining agents of forming max-size matchings across the domain, disqualifying it globally. Else, it is non-essential in every profile $\mathcal{L}_i(a)$ and by $b\in TC^*$ it is most preferred among the non-essential houses for any profile in $\mathcal{L}_i(a)$. Thus, $x_i(a)$ contains exactly the first choices and the subset of pseudo-second choices that remain feasible for a certainly popular matching.

    In Step 4, finding an $A$-perfect matching $M'$ in $G'$ that achieves the global maximum weight $\w(M') = \sum_{L \in \mathcal{L}} n_L$ strictly satisfies both the size constraint (Condition i) and the choice constraint (Condition ii) simultaneously across all layers. Because $x(a)$ correctly retains all usable pseudo-second houses, this restriction does not falsely eliminate any certainly popular matching. If no such matching $M'$ exists in $G'$, it is mathematically impossible to satisfy both conditions concurrently, and the algorithm correctly terminates.

    Since $D^{(i)}$ reachability and maximum weight verifications are executed a polynomial number of times (bounded by $\vert A \vert$, $\vert B \vert$, and $\vert P_a \vert$), the entire algorithm runs in polynomial time.
\end{proof}

\subsection{Popular common bases with integral comparison margins}
\label{subsec:integral-popular-common-base}

Kavitha et al.~\cite{yu2023arborescence} solve the popular common base problem with one-sided preferences in the intersection of a partition matroid and an arbitrary matroid when each partition class is equipped with a partial-order preference. To create the sufficient tools to solve the \ONETpopMLsum, we first extend their dual-certificate characterization and level algorithm to bounded integral comparison margins. The margins need not be induced by a partial order and may therefore contain cyclic comparisons.

While for our purposes, an extension in the bipartite matching model would be enough, we solve the matroidal case for the sake of generality.

Let
\[
    E=\mathbin{\dot\bigcup}_{a\in A}E_a,
\]
let $\mathcal P$ be the partition matroid with classes $(E_a)_{a\in A}$, and let
$\mathcal M=(E,\mathcal I)$ be a matroid of rank
\[
    r=|A|.
\]
For every $a\in A$, let
\[
    w_a:E_a\times E_a\longrightarrow\{-W,\ldots,W\}
\]
be integral and skew-symmetric, that is,
\[
    w_a(e,f)=-w_a(f,e)
    \qquad(e,f\in E_a).
\]
In particular, $w_a(e,e)=0$. If $F$ is a common base of $\mathcal P$ and $\mathcal M$, then $F(a)$ denotes the unique element of $F\cap E_a$. We call $F$ \emph{$w$-popular} if
\[
    \sum_{a\in A}w_a\bigl(F(a),F'(a)\bigr)\ge 0
\]
for every common base $F'$ of $\mathcal P$ and $\mathcal M$.

\begin{theorem}
\label{thm:integral-popular-common-base}
The existence of a $w$-popular common base can be decided using
    $O\bigl(r(r+1)W\bigr)$
calls to weighted matroid intersection and polynomial additional time in the instance size. If a $w$-popular common base exists, one can be found with the same running time. 
\end{theorem}

We first derive the dual certificates used by the algorithm. For a fixed common base $F$, define
\[
    \operatorname{wt}_F(e)=w_a\bigl(e,F(a)\bigr)
    \qquad(e\in E_a).
\]
For every common base $F'$,
\[\operatorname{wt}_F(F')
=\sum_{a\in A}w_a\bigl(F'(a),F(a)\bigr)\\
    =-\sum_{a\in A}w_a\bigl(F(a),F'(a)\bigr).
\]

Moreover, $\operatorname{wt}_F(F)=0$. Consequently, $F$ is $w$-popular if and only if the maximum weight of a common base under $\operatorname{wt}_F$ is zero.

Let $\rho$ be the rank function of $\mathcal M$. The common-base linear program and its dual are
\begin{align}\tag{$\mathrm P_F$} \max & \displaystyle\sum_{e\in E}\operatorname{wt}_F(e)x_e\\     &x(E_a)=1 \qquad(a\in A),\\     &x(S)\le \rho(S)\qquad(S\subseteq E),\\     &x_e\ge 0\qquad(e\in E), \end{align}
and
\begin{align}\tag{$\mathrm D_F$}
\min&\displaystyle\sum_{S\subseteq E}\rho(S)y_S
       +\sum_{a\in A}\alpha_a\\[1.5mm]
&\displaystyle\sum_{\substack{S\subseteq E\\e\in S}}y_S+\alpha_a
       \ge \operatorname{wt}_F(e)
       \qquad(a\in A,\ e\in E_a),\\
    &y_S\ge 0\qquad(S\subseteq E),\\
    &\alpha_a\text{ free}\qquad(a\in A).
\end{align}

A \emph{multichain} is an indexed sequence
\[
    \mathcal C=(C_1\subseteq C_2\subseteq\cdots\subseteq C_p)
\]
in which consecutive sets are allowed to be equal. When $C_p=E$, define
\[
    \operatorname{lev}_{\mathcal C}(e)
       =\min\{i:e\in C_i\}.
\]
For such a multichain, let
\begin{equation}
\label{eq:integral-admissible}
    E_w(\mathcal C)
    =\mathbin{\dot\bigcup}_{a\in A}
    \left\{
        e\in E_a:
        \operatorname{lev}_{\mathcal C}(e)
        -\operatorname{lev}_{\mathcal C}(f)
        \ge w_a(f,e)
        \text{ for every }f\in E_a
    \right\}.
\end{equation}

\begin{lemma}[Integral-margin dual certificates]
\label{lem:integral-certificate}
A common base $F$ is $w$-popular if and only if there is a multichain of flats
\[
    \mathcal D=(D_1\subseteq\cdots\subseteq D_q=E)
\]
such that
\begin{equation}
\label{eq:integral-certificate}
    F\subseteq E_w(\mathcal D)
    \qquad\text{and}\qquad
    \operatorname{span}_{\mathcal M}(F\cap D_i)=D_i
    \quad(1\le i\le q).
\end{equation}
Moreover, whenever $F$ is $w$-popular and $W\ge 1$, it has such a certificate with
\[
    q\le (r+1)W.
\]
\end{lemma}

\begin{proof}
Suppose first that $F$ is $w$-popular. Then the optimal values of
$\mathrm P_F$ and $\mathrm D_F$ are zero.

Among all optimal dual solutions, choose one minimizing
\[
    \sum_{S\subseteq E}y_S|S|\,|E\setminus S|.
\]
If two incomparable sets $S,T$ have positive multipliers, let
$\delta=\min\{y_S,y_T\}$ and replace
\[
\begin{array}{lll}
    y_S&\leftarrow y_S-\delta,&
    y_T\leftarrow y_T-\delta,\\
    y_{S\cap T}&\leftarrow y_{S\cap T}+\delta,&
    y_{S\cup T}\leftarrow y_{S\cup T}+\delta.
\end{array}
\]
For every element $e$, the total coefficient
$\sum_{R\ni e}y_R$ is unchanged. Submodularity of $\rho$ shows that the
dual objective does not increase. Since the original solution is optimal,
the modified solution is also optimal, and standard calculation shows
\[|S|\,|E\setminus S|+|T|\,|E\setminus T|
    -|S\cap T|\,|E\setminus(S\cap T)|
    -|S\cup T|\,|E\setminus(S\cup T)|
=2|S\setminus T|\,|T\setminus S|>0,
\]
contradicting the choice of the solution. Thus the positive $y$-variables
are supported on a chain.

Restricting $\mathrm D_F$ to this chain gives a constraint matrix whose
columns are incidence vectors of a chain and of the pairwise disjoint
partition classes $E_a$. This is the incidence matrix of the union of two
laminar families and is therefore totally unimodular. Since the right-hand sides are
integral, the restricted dual has an integral optimum. We may therefore
assume that all $y_S$ and $\alpha_a$ are integral.

Complementary slackness with the characteristic vector of $F$ gives
\begin{equation}
\label{eq:dual-cs-rank}
    |F\cap S|=\rho(S)
    \qquad\text{whenever }y_S>0
\end{equation}
and
\begin{equation}
\label{eq:dual-cs-element}
    \sum_{S\ni F(a)}y_S+\alpha_a=0
    \qquad(a\in A).
\end{equation}
For every support set $S$, replace $S$ by
$\operatorname{span}_{\mathcal M}(S)$. Its rank is unchanged and every
dual constraint can only become easier to satisfy. Furthermore,
\eqref{eq:dual-cs-rank} implies that $F\cap S$ is a basis of $S$.
Since $F$ is independent, no element of $F\setminus S$ lies in
$\operatorname{span}_{\mathcal M}(S)$. Hence
\[
    F\cap\operatorname{span}_{\mathcal M}(S)=F\cap S
\]
and
\[
    \operatorname{span}_{\mathcal M}
       \bigl(F\cap\operatorname{span}_{\mathcal M}(S)\bigr)
       =\operatorname{span}_{\mathcal M}(S).
\]
Closures preserve inclusion, so after merging equal sets the support is
still a chain, now consisting of flats. A positive multiplier on the empty
set may be deleted. We may also include $E$ in the support: increasing
$y_E$ by one and decreasing every $\alpha_a$ by one leaves every dual
constraint unchanged and changes the objective by
\[
    \rho(E)-|A|=r-r=0.
\]

We next bound the positive multipliers. Suppose that $y_S>W$ for a support
flat $S$. Decrease $y_S$ by one and increase $\alpha_a$ by one for each
$a$ satisfying $F(a)\in S$. The objective changes by
\[
    -\rho(S)+|F\cap S|=0.
\]
The constraint of an element $e\in E_a$ changes only when exactly one of
$e$ and $F(a)$ lies in $S$. If $F(a)\in S$ and $e\notin S$, its left-hand
side increases. Consider therefore $e\in S$ and $F(a)\notin S$. Write
\[
    L(x)=\sum_{R\ni x}y_R+\alpha_a
    \qquad(x\in E_a).
\]
Because the support is a chain, no support set can contain $F(a)$ without
also containing $e$: such a set would be incomparable with $S$. Hence
\[
    L(e)-L(F(a))
      =\sum_R y_R
       \bigl(\mathbf 1[e\in R]-\mathbf 1[F(a)\in R]\bigr)
      \ge y_S>W.
\]
By \eqref{eq:dual-cs-element}, $L(F(a))=0$, and therefore
$L(e)\ge W+1$. After decreasing $y_S$, the new left-hand side is at least
$W$, which is at least $\operatorname{wt}_F(e)$. Thus feasibility is
preserved. Repeating the operation gives an optimal certificate in which
\[
    1\le y_S\le W
\]
for every support flat.

A strict chain of flats contains at most $r+1$ sets: if $S\subsetneq T$
are flats, then $\rho(S)<\rho(T)$, and the possible ranks are
$0,1,\ldots,r$. Replace every support flat $S$ by $y_S$ consecutive
copies. This produces a multichain
\[
    \mathcal D=(D_1\subseteq\cdots\subseteq D_q=E)
\]
with
\[
    q\le (r+1)W.
\]
For every $x\in E$,
\[
    \sum_{S\ni x}y_S
      =|\{i:x\in D_i\}|
      =q-\operatorname{lev}_{\mathcal D}(x)+1.
\]
Using \eqref{eq:dual-cs-element}, the dual constraint for
$f\in E_a$ becomes
\[
    \operatorname{lev}_{\mathcal D}(F(a))
       -\operatorname{lev}_{\mathcal D}(f)
    \ge w_a\bigl(f,F(a)\bigr).
\]
This is precisely the condition $F(a)\in E_w(\mathcal D)$. The span
condition in \eqref{eq:integral-certificate} follows from
\eqref{eq:dual-cs-rank} and the fact that every $D_i$ is a flat.

Conversely, suppose that a common base $F$ and a multichain $\mathcal D$
satisfy \eqref{eq:integral-certificate}. For every distinct flat $S$
appearing in $\mathcal D$, let $y_S$ be its multiplicity, and set all
other $y$-variables to zero. Define
\[
    \alpha_a=-|\{i:F(a)\in D_i\}|.
\]
For $f\in E_a$, the corresponding dual left-hand side is
\[
    |\{i:f\in D_i\}|-|\{i:F(a)\in D_i\}|
    =
    \operatorname{lev}_{\mathcal D}(F(a))
      -\operatorname{lev}_{\mathcal D}(f)\\
    \ge w_a\bigl(f,F(a)\bigr)
     =\operatorname{wt}_F(f),
\]
where the inequality follows from $F(a)\in E_w(\mathcal D)$. Thus the
dual solution is feasible. Its objective value is
\[
\sum_{i=1}^{q}\rho(D_i)
       -\sum_{a\in A}|\{i:F(a)\in D_i\}|
    =
    \sum_{i=1}^{q}\bigl(\rho(D_i)-|F\cap D_i|\bigr)\\
    =0.
\]
The primal optimum is therefore at most zero. Since $F$ itself has weight
zero, the optimum is zero and $F$ is $w$-popular.
\end{proof}

\subsubsection{The level algorithm}
Set
\[
    L=(r+1)W.
\]
For $W\ge1$, the algorithm maintains a multichain of flats
\[
    \mathcal C=(C_1\subseteq\cdots\subseteq C_p=E).
\]
It starts with $p=1$ and $C_1=E$ and repeats the following steps.

\begin{enumerate}
\item Compute the admissible set $E_w(\mathcal C)$.

\item Among all common independent sets
\(
    I\subseteq E_w(\mathcal C)
\)
of $\mathcal P$ and $\mathcal M$, find one that lexicographically
maximizes
\[
    \bigl(|I\cap C_1|,\ldots,|I\cap C_p|\bigr).
\]

\item If
\[
    |I\cap C_i|=\rho(C_i)
    \qquad(1\le i\le p),
\]
return $I$.

\item Otherwise, let $j$ be the smallest index satisfying
\[
    |I\cap C_j|<\rho(C_j)
\]
and replace
\[
    C_j\leftarrow
       \operatorname{span}_{\mathcal M}(I\cap C_j).
\]
If $j=p$, append a new last set $C_{p+1}=E$ and increase $p$ by one.

\item If $p>L$, report that no $w$-popular common base exists.
\end{enumerate}

The lexicographic optimization in Step~2 can be handled with weighted matroid-intersection, e.g., using weights $$ \lambda_{\mathcal C}(e)
       =\sum_{\substack{1\le i\le p\\e\in C_i}}(r+1)^{p-i}.$$ 

The update preserves the multichain property. For every $i<j$, the
minimality of $j$ gives
\[
    |I\cap C_i|=\rho(C_i).
\]
As $C_i$ is a flat,
\[
    \operatorname{span}_{\mathcal M}(I\cap C_i)=C_i.
\]
Consequently, when $j>1$,
\[
    C_{j-1}
      \subseteq
      \operatorname{span}_{\mathcal M}(I\cap C_j).
\]
The updated set is contained in the old $C_j$, because the latter is a
flat, and hence remains contained in $C_{j+1}$. If $j=p$, its rank
strictly decreases, so the updated $C_p$ is a proper subset of $E$ and a
new copy of $E$ can be appended.

\begin{lemma}[Level transfer]
\label{lem:integral-level-transfer}
Let
\(
    \mathcal C=(C_1\subseteq\cdots\subseteq C_p=E)
    \quad\text{and}\quad
    \mathcal D=(D_1\subseteq\cdots\subseteq D_q=E)
\)
satisfy
\begin{equation}
\label{eq:integral-invariant}
    p\le q
    \qquad\text{and}\qquad
    D_i\subseteq C_i
    \quad(1\le i\le p).
\end{equation}
Then, the following statements hold:
\begin{enumerate}
\item for every $x\in E$, $ \operatorname{lev}_{\mathcal D}(x)
       \ge \operatorname{lev}_{\mathcal C}(x);$

\item if $f\in E_w(\mathcal D)$ and
$ \operatorname{lev}_{\mathcal C}(f)
      =\operatorname{lev}_{\mathcal D}(f),$ 
then $f\in E_w(\mathcal C)$;

\item if $e,f\in E_a$, $e\in E_w(\mathcal C)$,
$f\in E_w(\mathcal D)$, and
\(\operatorname{lev}_{\mathcal C}(f)
      =\operatorname{lev}_{\mathcal D}(f),
\)
then
\( \operatorname{lev}_{\mathcal C}(e)
      =\operatorname{lev}_{\mathcal D}(e).
\)
\end{enumerate}
\end{lemma}

\begin{proof}
The first statement follows immediately from $D_i\subseteq C_i$.

For the second, let
\[
    t=\operatorname{lev}_{\mathcal C}(f)
      =\operatorname{lev}_{\mathcal D}(f).
\]
Assume $f\in E_a$. For every $g\in E_a$,
\[
\operatorname{lev}_{\mathcal C}(f)
       -\operatorname{lev}_{\mathcal C}(g)
    \ge
    \operatorname{lev}_{\mathcal D}(f)
       -\operatorname{lev}_{\mathcal D}(g)\\
    \ge w_a(g,f).
\]
Hence $f\in E_w(\mathcal C)$.

For the third, write
\[
    x=\operatorname{lev}_{\mathcal C}(e),\qquad
    y=\operatorname{lev}_{\mathcal D}(e),\qquad
    t=\operatorname{lev}_{\mathcal C}(f)
      =\operatorname{lev}_{\mathcal D}(f).
\]
Admissibility gives
\[
    x-t\ge w_a(f,e)
\quad \text{and}\quad
    t-y\ge w_a(e,f)=-w_a(f,e).
\]
Thus $x\ge y$. The first statement gives $y\ge x$, so $x=y$.
\end{proof}

We use the following standard strong exchange property of matroids (see \cite{brualdi1969comments}): if
$X,Y\in\mathcal I$, $e\in X\setminus Y$, and $Y+e\notin\mathcal I$, then
there is an $f\in Y\setminus X$ such that
\[
    X-e+f\in\mathcal I
    \qquad\text{and}\qquad
    Y+e-f\in\mathcal I.
\]

\begin{lemma}[Update invariant]
\label{lem:integral-update}
Let $F$ be a $w$-popular common base with a certificate
\(
    \mathcal D=(D_1\subseteq\cdots\subseteq D_q=E)
\)
satisfying Lemma~\ref{lem:integral-certificate}. During the level
algorithm, the current multichain
\(
    \mathcal C=(C_1\subseteq\cdots\subseteq C_p=E)
\)
always satisfies \eqref{eq:integral-invariant}.
\end{lemma}

\begin{proof}
Initially $\mathcal C=(E)$, so the assertion holds. Let 
\[C_0=D_0:=\emptyset\]

Assume that \eqref{eq:integral-invariant} holds before an iteration. Let
$I$ be the lexicographically maximum common independent set found by the
algorithm, let $j\ge 1$ be the first deficient index, and put
\[
    C'_j=\operatorname{span}_{\mathcal M}(I\cap C_j).
\]
We prove
$$ D_j\subseteq C'_j.$$

Suppose not. Since
\(
    \operatorname{span}_{\mathcal M}(F\cap D_j)=D_j,
\)
there is an element
\(
    f_1\in F\cap D_j
       \setminus\operatorname{span}_{\mathcal M}(I\cap C_j).
\)
For every $i<j$,
\(  C_i=\operatorname{span}_{\mathcal M}(I\cap C_i).
\)
Hence, 
\(
    D_{j-1}\subseteq C_{j-1}
       \subseteq\operatorname{span}_{\mathcal M}(I\cap C_j),
\)
so $f_1\notin D_{j-1}$ and $f_1\notin C_{j-1}$. Therefore
\[
    \operatorname{lev}_{\mathcal C}(f_1)
      =\operatorname{lev}_{\mathcal D}(f_1)=j.
\]
Since $f_1\in F\subseteq E_w(\mathcal D)$,
Lemma~\ref{lem:integral-level-transfer} gives
$f_1\in E_w(\mathcal C)$.

The set $
    J_1=(I\cap C_j)+f_1$
is independent in $\mathcal M$ and lexicographically better than $I$.
Therefore it is not independent in the partition matroid. Hence there is
an element $e_1\in I\cap C_j$ in the same partition class as $f_1$.
Since $f_1$ is the unique element of $F$ from that class, $e_1\notin F$.
Since
\(
    e_1\in I\subseteq E_w(\mathcal C)
\) Lemma~\ref{lem:integral-level-transfer}(3) gives
\(
    \operatorname{lev}_{\mathcal C}(e_1)
       =\operatorname{lev}_{\mathcal D}(e_1).
\)
Moreover, $e_1\in C_j$, and hence
\[
    \operatorname{lev}_{\mathcal C}(e_1)
       =\operatorname{lev}_{\mathcal D}(e_1)\le j.
\]

We now inductively construct distinct elements
$f_1,f_2,\ldots$ of $F$. Suppose that $J_t$ is an
$\mathcal M$-independent subset of $E_w(\mathcal C)$, lexicographically
better than $I$, and has a unique partition conflict between elements
$e_t\notin F$ and $f_t\in F$ in the same class, where
\[
    \operatorname{lev}_{\mathcal C}(e_t)
       =\operatorname{lev}_{\mathcal D}(e_t)\le j.
\]
We also maintain the invariant
\begin{equation}
\label{eq:J-prefix-counts}
    |J_t\cap C_i|=|I\cap C_i|
    \qquad(1\le i<j).
\end{equation}
This invariant holds for $J_1$: indeed, $f_1$ has level $j$ in
$\mathcal C$, and therefore $
    J_1\cap C_i=I\cap C_i
    \; (1\le i<j).$

Apply the strong exchange property in $\mathcal M$ to $J_t$, $F$, and $e_t$. Since
$F$ is a base, there is
\(
    f_{t+1}\in F\setminus J_t
\)
such that
\[
    J_{t+1}=J_t-e_t+f_{t+1}\in\mathcal I
\quad\text{and}\quad
    F+e_t-f_{t+1}\in\mathcal I.
\]

We claim that
\[
    \operatorname{lev}_{\mathcal C}(e_t)
    \le \operatorname{lev}_{\mathcal C}(f_{t+1})
    \le \operatorname{lev}_{\mathcal D}(f_{t+1})
    \le \operatorname{lev}_{\mathcal D}(e_t).
\]
The middle inequality follows from
Lemma~\ref{lem:integral-level-transfer}(1).

For the last inequality, let
$\ell=\operatorname{lev}_{\mathcal D}(f_{t+1}).$
If $\ell=1$, then
\(
    \operatorname{lev}_{\mathcal D}(e_t)\ge 1=\ell
\)
is immediate. Suppose therefore that $\ell\ge2$. By the definition of
the level,
\(
    f_{t+1}\notin D_{\ell-1},
\)
and hence
\[
    F\cap D_{\ell-1}\subseteq F-f_{t+1}.
\]
Using the certificate property and monotonicity of span, we obtain
\[
    D_{\ell-1}
       =\operatorname{span}_{\mathcal M}(F\cap D_{\ell-1})
       \subseteq
       \operatorname{span}_{\mathcal M}(F-f_{t+1}).
\]
On the other hand, the independence of
$F+e_t-f_{t+1}$ implies
$   e_t\notin
       \operatorname{span}_{\mathcal M}(F-f_{t+1}).
$

Consequently $e_t\notin D_{\ell-1}$, and therefore
\[
    \operatorname{lev}_{\mathcal D}(e_t)\ge\ell
       =\operatorname{lev}_{\mathcal D}(f_{t+1}).
\]

For the first inequality, let
$h=\operatorname{lev}_{\mathcal C}(e_t).$

If $h=1$, the inequality is immediate. Suppose that $h\ge2$ and, for a
contradiction, that
$\operatorname{lev}_{\mathcal C}(f_{t+1})<h.$
    
Then $f_{t+1}\in C_{h-1}$. Since $h\le j$, we have $h-1<j$, so the
minimality of $j$ and \eqref{eq:J-prefix-counts} give
\[
    |J_t\cap C_{h-1}|
       =|I\cap C_{h-1}|
       =\rho(C_{h-1}).
\]
Because $J_t$ is independent, $J_t\cap C_{h-1}$ is therefore a basis
of $C_{h-1}$. Hence
\[
    C_{h-1}
      =\operatorname{span}_{\mathcal M}(J_t\cap C_{h-1}).
\]
Moreover, $e_t\notin C_{h-1}$ by the definition of $h$, and thus
\[
    J_t\cap C_{h-1}\subseteq J_t-e_t.
\]
It follows that
\[
    C_{h-1}
      =\operatorname{span}_{\mathcal M}(J_t\cap C_{h-1})
      \subseteq\operatorname{span}_{\mathcal M}(J_t-e_t).
\]
Since $f_{t+1}\in C_{h-1}$, this contradicts the independence of
$J_t-e_t+f_{t+1}$. Therefore
\[
    \operatorname{lev}_{\mathcal C}(e_t)
       \le\operatorname{lev}_{\mathcal C}(f_{t+1}).
\]

Because
$\operatorname{lev}_{\mathcal C}(e_t)
      =\operatorname{lev}_{\mathcal D}(e_t),$
all inequalities are equalities. Thus
\[
    \operatorname{lev}_{\mathcal C}(f_{t+1})
      =\operatorname{lev}_{\mathcal D}(f_{t+1}),
\]
and Lemma~\ref{lem:integral-level-transfer}(2) gives
$f_{t+1}\in E_w(\mathcal C)$.

Replacing $e_t$ by $f_{t+1}$ preserves the number of elements in every
$C_i$, so $J_{t+1}\subseteq C_j$ remains lexicographically better than $I$. It cannot
be partition-independent, by the choice of $I$. Hence $f_{t+1}$ creates
a new partition conflict with some $e_{t+1}\notin F$ in the same class.
Again,
\[
    \operatorname{lev}_{\mathcal C}(e_{t+1})
       =\operatorname{lev}_{\mathcal D}(e_{t+1})\le j.
\]
Moreover, $f_{t+1}\notin J_t$, so it is distinct from
$f_1,\ldots,f_t$ (earlier $f_i$'s are never removed, as only elements $e_i\notin F$ are removed). This produces arbitrarily many distinct elements of
the finite base $F$, a contradiction. Therefore
\[
    D_j\subseteq C'_j.
\]

Thus the update preserves $D_i\subseteq C_i$ for all current indices. If
$j=p$, then the updated $C_p$ is a proper subset of $E=D_q$. Since
$D_p\subseteq C_p$, we have $p<q$, so after appending a new copy of $E$
the inequality $p+1\le q$ remains valid.
\end{proof}

\begin{proof}[Proof of Theorem~\ref{thm:integral-popular-common-base}]
If $W=0$, all margins vanish, and ordinary matroid intersection suffices.
Assume $W\ge1$.

If the algorithm returns $I$, then $C_p=E$ and
\[
    |I|=|I\cap E|=\rho(E)=r.
\]
Thus $I$ is a common base. Moreover,
\[
    \operatorname{span}_{\mathcal M}(I\cap C_i)=C_i
\]
for every $i$, and $I\subseteq E_w(\mathcal C)$. By
Lemma~\ref{lem:integral-certificate}, $I$ is $w$-popular.

Conversely, suppose that a $w$-popular base $F$ exists. By
Lemma~\ref{lem:integral-certificate}, it has a certificate of length
\[
    q\le (r+1)W.
\]
Lemma~\ref{lem:integral-update} gives $p\le q$ throughout the algorithm.
Hence the algorithm cannot create more than $(r+1)W$ levels while a
$w$-popular base exists. Reporting nonexistence when this bound is
exceeded is therefore correct.

To bound the number of iterations, set
\[
    L=(r+1)W
\]
and define $C_i=E$ for undefined indices $p<i\le L$. Consider
\[
    \Phi(\mathcal C)=\sum_{i=1}^{L}\rho(C_i).
\]
Initially, $\Phi=Lr$. In every iteration, the deficient set $C_j$ is
replaced by
$\operatorname{span}_{\mathcal M}(I\cap C_j),$
whose rank is
\[
    |I\cap C_j|<\rho(C_j).
\]
Thus $\Phi$ decreases by at least one. Appending a copy of $E$ does not
increase $\Phi$, because undefined indices were already counted as $E$.
Hence the number of iterations, and therefore the number of weighted
matroid-intersection calls, is at most
\[
    Lr=r(r+1)W.
\]

All remaining operations, including computing levels, admissible
elements, ranks, and spans, take polynomial time. This proves the theorem.
\end{proof}

\subsection{Aggregated preferences}\label{subsec:aggr}

For two strict profiles, the sign of each nonzero aggregate comparison can be represented by a partial order. This shortcut fails from three profiles onward: the orders $b_1\succ b_2\succ b_3$, $b_2\succ b_3\succ b_1$, and $b_3\succ b_1\succ b_2$ induce the majority cycle $b_1\succ b_2\succ b_3\succ b_1$. The integral-margin problem above was designed precisely to retain these cyclic comparisons and their different magnitudes.

\begin{theorem}\label{thm:ONEsidedsum-general}
The problems \ONEpopMLsum\ and \ONEpopunsum\ can be solved in polynomial time. The same result holds when the input preference lists are weak orders. It can also be verified in polynomial time whether a given matching is sum-popular.
\end{theorem}
\begin{proof}
We already observed in the preliminaries that the multilayer and uncertainty aggregation formulations are equivalent, so consider the multilayer formulation. For every applicant $a$ and two incident edges $e,f$, define
\[
    w_a(e,f)=\sum_{i=1}^k \vote_a^{L_i}(e,f).
\]
For the uncertainty formulation, the sum is instead over the $k$ lists in $P_a$. The function $w_a$ is integral and skew-symmetric, and $|w_a(e,f)|\le k$; a tie contributes zero.

Append a private last-resort house to every applicant, ranked last in every profile. Let $r=|A|$, let $E=\mathbin{\dot\bigcup}_{a\in A}\delta(a)$ be the resulting edge set, and let $\mathcal M$ be the rank-$r$ truncation of the partition matroid that enforces the house capacities. An $A$-complete matching is precisely a common base of $\mathcal M$ and the partition matroid with classes $\delta(a)$, $a\in A$. For any two such matchings $M,N$,
\[
    \sum_{a\in A}w_a(M(a),N(a))=\Delta^{ML}(M,N).
\]
Thus $M$ is sum-popular if and only if it is a $w$-popular common base in the sense of Theorem~\ref{thm:integral-popular-common-base}. Here $W\le k$, and the $k$ profiles are explicitly listed, so that theorem gives a polynomial-time algorithm.

For verification, extend the given matching $M$ with its private last-resort edges and assign every edge $e=(a,b)$ the weight
\[
    c_M(a,b)=\sum_{i=1}^k\vote_a^{L_i}(M(a),b).
\]
The weight of every $A$-complete matching $N$ is exactly $\Delta^{ML}(M,N)$. Hence $M$ is sum-popular if and only if the minimum such weight is nonnegative, which can be checked by a minimum-weight capacitated matching algorithm. The same argument applies to weak orders because ties contribute zero.
\end{proof}

\section{Two-sided preferences}\label{sec:twosided}
In this section, we investigate the problems for two-sided preferences. First, in Subsection~\ref{subsec:verif} we show that verifying if a given matching is certainly popular or certainly dominant can be done in polynomial time. Then, we study the existence questions, first for dominant matchings in Subsection~\ref{subsec:dom}, and then for popular matchings for Subsection~\ref{subsec:pop}.
We give a thorough complexity picture for these problems. Our algorithms provide efficient methods of computing desirable, certainly popular solutions in such uncertain or dynamic conditions. On the other hand, our hardness results highlight the boundaries of tractability and show the cases where we need heuristics, integer programming or other machinery to tackle these problems. 
\subsection{Verification is solvable}\label{subsec:verif}
In this subsection, we show that deciding if a matching $M$ is certainly popular; sum-popular; or robust popular can be done in polynomial time. 

\begin{theorem}
\label{thm:2sideverif}
   Verifying if a matching $M$ is certainly popular (resp. dominant), sum-popular (resp. dominant) or robust-popular (resp. dominant) can be done in polynomial-time in all of \popML, \domML, \popun, \domun, \popunsum, \domunsum, \poprob\ and \domrob. 
   
\end{theorem}
\begin{proof}

The case of \popML\ and \domML\ is trivial, because we only need to verify that $M$ is popular (resp. dominant) in polynomially many preference profiles in the input size, and verifying if a matching is popular can be done in polynomial-time~\cite{huang2013popular}. 

    We may always assume without loss of generality that $|A|=|B|$. Otherwise, we can add isolated agents to the instance, whose vote will always be 0 when comparing any two matchings and the other agents' votes are unaffected.

    Let $I$ be an instance of \popun\ or \poprob\ with $G=(A,B,E)$ and $\mathcal{L}$ denote the set of possible preference profiles. Let $M$ be a matching. First, we extend $G$ to be a complete bipartite graph by adding all missing edges. Call this new graph $H$. Then, we create a weight function $w:E(H)\to \mathbb{Z}$ as follows. For each $e=(a,b)\in E$, we let $w(e)=\vote^*_a(M(a),e)+\vote^*_b(M(b),e)$. Here, $M(a)$ is the edge incident to $a$, if there is any, otherwise it is $\emptyset$. Furthermore, $\vote_a^*(g,f)$ is defined to be $-1$, if there is a possible preference list, where $a$ prefers his partner in $f$, $0$ if $g=f$ and $+1$, if $a$ always prefers his partner in $g$ (so in the $k$-robust case, this corresponds to that partner being at least $(k+1)$-better). For an edge $e=(a,b)\notin E$, we let $w(e)=\vote_a^*(M(a),\emptyset )+\vote_b^*(M(b),\emptyset )$  (here $\vote_a^*(M(a),\emptyset )=+1$, except the case when $M(a)=\emptyset$, when it is $0$).

    We claim that $M$ is popular, if and only the minimum weight perfect matching in $H$ has weight at least $0$. Indeed, for any perfect matching $N'$ of $H$, if we let $N=N'\cap E$, then $w(N')$ is exactly $\min_{L\in \mathcal{L}}\Dvote^L(M,N)$ by the definition of the weights and the fact that since $N'$ is perfect, we count the smallest possible vote for every agent in $A\cup B$ exactly once. Hence, $M$ is popular, if and only if there is no perfect matching $N'\subset E$ such that $w(N')<0$. This can be verified in polynomial-time with various methods, e.g., with the Hungarian method~\cite{kuhn1955hungarian}. 
    
    For the case of \domun\ and \domrob, we can decide similarly if $M$ is certainly popular, and if not, then conclude that it is not certainly dominant. To decide whether there is a larger matching $N$, such that $\Dvote (M,N)=0$ in some possible preference profile, we subtract $\varepsilon=\frac{1}{|M|+1}$ from each original edge's weight and check whether there is a perfect matching with weight at most $-1$. If there is a perfect matching $N'$ of weight at most $-1$, then for $N=N'\cap E$, we have that $|N|\ge |M|+1$, as $|M|$ is popular. Also, $|N|\le 2|M|$, because any popular matching is (inclusionwise) maximal. Hence, the weight of $N$ cannot be $1$ or more without subtracting the $|N|$ many $\varepsilon$ values, so it is exactly $0$. Hence, $N$ is larger than $M$, but $\Dvote (M,N)=0$ in some possible preference profile, so $M$ is not certainly dominant. In the other direction, if $M$ is not certainly dominant, then there is a matching $N$ with $|N|\ge |M|+1$ and $\Dvote (M,N)=0$ in some possible preference profile, so $N$ can be extended to a perfect matching $N'$ with weight at most $-1$. 

    Finally, let $I$ be an instance of \popunsum. Let $\mathcal{L}=\{ L_1,\dots,L_k\}$ be the set of preference profiles. Again, we create $H$ by making the graph $G$ a complete bipartite graph. Then, we create a weight function $w:E(H)\to \mathbb{Z}$ as follows. If $e=(a,b)\in E$, then $w(e)=\sum_{L\in \mathcal{L}}(\vote_a^L(M(a),e)+\vote_b^L(M(b),e) )$. If $e=(a,b)\notin E$, then $w(e)=\sum_{L\in \mathcal{L}}(\vote_a^L(M(a),\emptyset )+\vote_b^L(M(b),\emptyset ) )$. 

    Again, the problem of deciding whether $M$ is sum-popular, becomes equivalent to deciding if there is a perfect matching $N'$ in $H$ such that $w(N')<0$. For \domunsum, we can check dominance too with the same trick of subtracting $\frac{1}{|M|+1}$ from the weight of the original edges.
    \end{proof}

\subsection{Dominant Matchings}\label{subsec:dom}

We start by discussing our results for the two-sided markets in the case of dominant matchings. Our first result is the hardness of \domML. 

\begin{theorem}
\label{thm:domML}
    \domML\ is NP-hard, even for a fixed constant $k=2$.
\end{theorem}
\begin{proof}
    We reduce from the NP-complete problem \textsc{2-stable-multilayer}, which consists of finding a certainly stable matching in a bipartite graph $G=(A,B,E)$ in two preference profiles $L_1,L_2$. This is shown to be NP-complete by \cite{miyazaki2017jointly}. We can also suppose that the graph is a complete bipartite graph, by adding the remaining edges and extending the preferences in $L_1$ and $L_2$ by putting every new neighbor of an agent to the end of his preference list in the same arbitrary order in both $L_1$ and $L_2$. 

    Let $I$ be an instance of \textsc{2-stable-multilayer}. We create an instance $I'$ of \domML. Suppose by symmetry that $n=|A|\le |B|=m$. Then, we add $n$ more agents $a_1',\dots,a_n'$, such that $a_i'$ only considers $a_i$ acceptable and $a_i'$ is the worst in $a_i$'s preference list in both $L_1$ and $L_2$. Otherwise, we just keep the same preferences for each agent to obtain $L_1'$ and $L_2'$. 

    We claim that there is a matching $M$ that is certainly stable in $I$ if and only if there is a certainly dominant matching in $I'$.

    Let $M$ be a certainly stable matching in $I$. Then, $M$ matches every agent in $A$, as the preference lists are complete. Therefore, $M$ is also certainly stable in $I'$, as no $(a_i,a_i')$ edge blocks it. Hence $M$ is certainly popular, as stable matchings are always popular. \cite{gardenfors1975match} Moreover, as one side of the bipartite graph in $I'$ has size $n$, it follows that $M$ is a maximum size matching, therefore $M$ is certainly dominant too. 

    For the other direction, suppose that there is a certainly dominant matching $M$. Observe that no edge $(a_i,a_i')$ can be included in $M$, as otherwise there would be an unmatched agent $b_j\in B$, so $M\setminus \{ (a_i,a_i')\} \cup \{ (a_i,b_j\}$ would dominate $M$ in both profiles, contradiction ($a_i$ and $b_j$ improve and only $a_i'$ geets worse). Hence, $M$ is a matching in $I$ too. Suppose that there is a blocking edge $(a_i,b_j)$ to $M$ in $L_i$. Let $M(a_i)=b_k$, $M(b_j)=a_l$. Then, $M\setminus \{ (a_i,b_k),(a_l,b_j)\} \cup \{ (a_i,b_j),(a_l,a_l')\}$ dominates $M$ in $L_i'$, as $a_i,b_j,a_l'$ all improve and only $a_l,b_k$ get worse,  contradiction. 
\end{proof}

We proceed by describing a well known reduction that connects dominant and stable matchings. Let $I=(G,L)$, $G=(A,B,E)$ be an instance of \pop. We create an instance $I'$, by adding two parallel edges $x(e),y(e)$ for each edge $e\in E$. Then, we create the preference profile in $I'$ (over the edges now, as we have parallel edges in the new instance) from $L$, by the following rule. For each agent $a\in A$, we first rank all the edges of type $x$ incident to it, in the same order as $a$'s preference over his neighbors, then we rank all incident edges of type $y$, again in the order of $a$'s preference list. So if $a$ had preference list $b_1\succ b_2$, then it becomes $x(b_1)\succ x(b_2)\succ y(b_1) \succ y(b_2)$. For the agent $b\in B$, we do it similarly, by interchanging the order of the $x$ and $y$ copies. The following theorem is known.

\begin{theorem}[\cite{kavitha2014size,kamiyama2020popular}]
    \label{thm:connection}
    Every stable matching $M'$ in $I'$ gives a dominant matching $M=f(M')$ in $I$, where $f(M')=\{ e\mid |M'\cap \{ x(e),y(e)\} |\ge 1\}$. Furthermore, for every dominant matching $M$ in $I$, there is a matching $M'$ in $I'$ such that $M=f(M')$.
\end{theorem}

The only problem with Theorem \ref{thm:connection} is that for a dominant matching $M$, the stable matching $M'$ with $f(M')=M$ also depends on the preference profile in $I$. Hence, finding a certainly dominant matching does not reduce trivially to finding a certainly stable matching, as for a certainly dominant matching, its preimage in $I'$ may be different for each preference profile, so there may not be a certainly stable matching in $I'$ that gives $M$.

We show that we can still utilize this connection to the stable matching problem to solve \domun\ and \domrob. 

\begin{theorem}
\label{thm:domun}
    \domun\ and \domrob\ can be solved in polynomial-time. 
\end{theorem}
\begin{proof}
We have already seen that it is enough to give an algorithm for \domun.

Let $I$ be an instance of \domun. We create an instance $I'$ of \textsc{stable-uncertain}, where for each edge $e$, we create two copies $x(e)$ and $y(e)$ as in Theorem \ref{thm:connection}, and for each agent $u$, his possible preference lists are created from his original preference lists $\succ_u^i$, such that if $u\in A$, then we first rank the $x$ copies of the edges according to $\succ_u^i$ and then the $y$ copies in the same order, and if $u\in B$, then we rank the $y$ copies first and the $x$ copies last. 

By Theorem \ref{thm:connection}, we get that if there is a certainly stable matching in $I'$, then that gives us a certainly dominant matching $M$ in $I$. Therefore, our goal is to show that if there is a certainly dominant matching $M$ in $I$, then there is a certainly stable matching $M'$ in $I'$. For this, we first create a strict preference list $\succ_u'$ for each $u\in A\cup B$, such that for any $v\in V(u)$, if $v\succ_u^iM(u)$ for some possible preference list $\succ_u^i$ of $u$, then $v\succ_u'M(u)$, and if $M(u)\succ_u^iv$ for all possible preference lists, then $M(u)\succ_u'v$. Otherwise, the preferences between the agents that are both better or both worse than $M(u)$ are created in an arbitrary way. As $M$ is certainly dominant, it is also dominant with the preferences $\succ_u'$, $u\in A\cup B$, because $\min_{i\in k}\{\vote_{\succ_u^i}(M,N)\}=\vote_{\succ_u'}(M,N)$ for any matching $N$. 

By Theorem \ref{thm:connection}, this matching $M$ gives a stable matching $M'$ with respect to the extensions of the $\succ_u'$ preferences. We claim that it is certainly stable in $I’$. Suppose that an edge $x(a,b)$ blocks $M'$ with the extension of some preference lists of $a$ and $b$. Then, if $a$ is not matched by an $x$ copy, $a$ prefers this edge in the extension of $\succ_a'$ too. Also, $b$ must be unmatched or matched with a $x$-type copy and in the latter case, there must be a preference list $\succ_b^i$, such that $a\succ_b^iM(b)$. However, then $x(a,b)\succ_b'M'(b)$, contradicting that $M'$ is stable. If $a$ is matched with an $x$-type copy, then the fact that $x(a,b)$ blocks imply that $b\succ_a^iM(a)$ for some $i$ and so $x(a,b)\succ_a'M'(a)$. Similarly as before, we get that $x(a,b)\succ_b'M'(b)$ should also hold and therefore $M'$ is not stable, contradiction. The case for edges of type $y$ is analogous. 

Hence, we conclude that $M'$ is certainly stable, and therefore the theorem follows. 
\end{proof}

\subsection{Popular matchings}\label{subsec:pop}

We continue the section with our results for popular matchings.
We state a simple observation that we are going to use many times in our proofs without further reference. 

\begin{observation}
If $a,b$ are best in each other preference lists, then in any popular matching $M$, it holds that both $a$ and $b$ are covered by $M$.    
\end{observation}
Indeed, if one of them is unmatched, then by coming together, both $a$ and $b$ would strictly improve, and at most one agent gets worse off. 

\begin{theorem}
\label{thm:popML}
    \popML\ and \popun\ are NP-complete, even if uncertainty occurs only on one side and a fixed constant $k=2$.
\end{theorem}

\begin{proof}
The problems are in NP by Theorem \ref{thm:2sideverif}.

We reduce from the NP-complete problem \popf.
Let $I$ be an instance of \popf. We create an instance $I'$ of \popML\ and an instance $I''$ of \popun\ with $k=2$ as follows. 

Let $G=(A,B;E)$ denote the underlying bipartite graph of the instance $I$, where $A=\{ a_1,\dots,a_n \}$ and $B=\{ b_1,\dots,b_m\}$. Let $(x,y)\in E$ be the forced edge with $x\in A$ and $y\in B$. By using the restrictions and by reindexing, we can assume that $x=a_n,y=b_m$ are each other's best choice, $V(x)=\{ a_1,y\}$ and that the forced vertex is some $z\in B\setminus \{ a_1,y\}$. We can also suppose by reindexing that the neighbors of $y$ are $x,b_1,\dots,b_l$. 

In $I'$ and $I''$, we keep all vertices of $G$. Furthermore, we add vertices $a_1',z',d_z$ and vertices $d_{a},d_{b}$. 

For $I'$, we create the preferences profiles in the two layers as follows. In profile $L_1$:
\begin{center}

\begin{tabular}{ll}
  $a_1:$ & $d_{a}\succ x$   \\
  $d_b:$ & $d_a\succ b_1\succ \cdots \succ b_l$ \\
  $a_1':$  & $d_a$ \\
  $z:$ & $[V(z)] \succ d_z$   \\
  $z':$ & $d_z$ \\
  \hline
   $b_i:$  & $[V(b_i)] \succ d_b$ \\
   $d_a:$ & $d_b\succ a_1' \succ a_1$ \\
   $d_z:$ & $z\succ z'$\\

\end{tabular}

\end{center}

Here, $[V(z)]$ and $[V(b_i)]$ denote the neighbors of $z$ and $b_i$ ranked according to their original orders, respectively.
For every other agent, their preferences are unchanged from $I$ in $L_1$.

In profile $L_2$, we have: 
\begin{center}
\begin{tabular}{ll}
  $a_1:$ & $d_{a}\succ x $ \\
  $d_b:$ & $d_a\succ  b_1\succ \cdots \succ b_l$ \\
  $a_1':$  & $d_a$ \\
  $z:$ & $[V(z)] \succ d_z$   \\
  $z':$ & $d_z$ \\
  \hline
   $b_i:$  & $[V(b_i)\setminus y] \succ d_b \succ y$  \\
   $d_a:$ & $d_b\succ a_1 \succ a_1'$ \\
   $d_z:$ & $z'\succ z$

\end{tabular}
\end{center}
For every other agent, their preferences are unchanged from $I$ in $L_2$.

For $I''$, we define the sets $P_u$ for each agent $u$ as the union of their two preference lists in $L_1$ and $L_2$.

It is straightforward to verify that each agent on the side $A$ has the same preferences in both $L_1$ and $L_2$, so uncertainty occurs on one side of the instance only. 

The reduction is illustrated in Figures \ref{fig:orig}, \ref{fig:L1} and \ref{fig:L2} for a given instance $I$. 

\begin{figure*}
    \centering
    \includegraphics{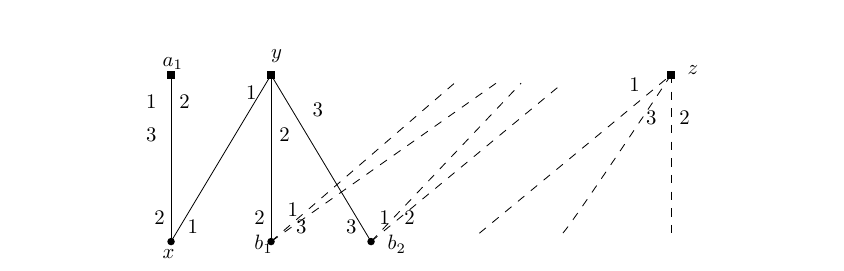}
    \caption{An instance $I$ of \popf\ with forced edge $(x,y)$ and forced vertex $z$. The numbers over the edges indicate the preferences.}
    \label{fig:orig}
\end{figure*}

\begin{figure*}
    \centering
    \includegraphics{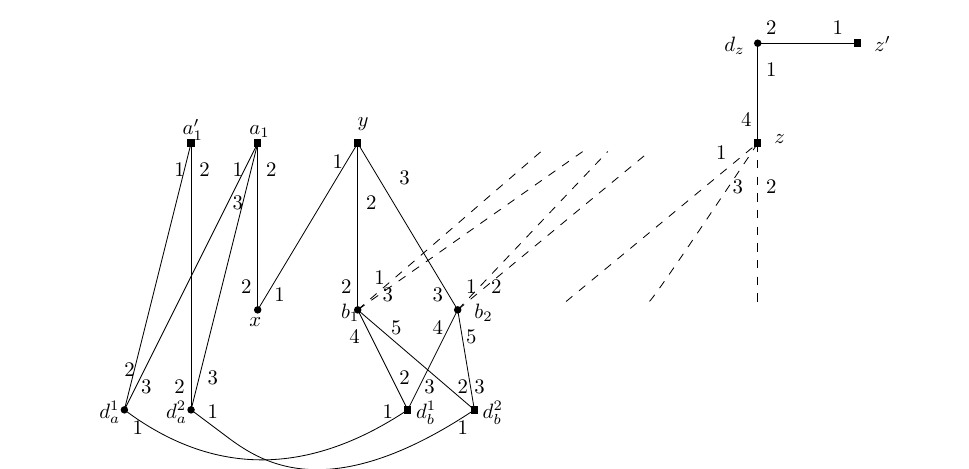}
    \caption{The instance $I'$. The numbers illustrate the preferences in the preference profile $L_1$ created from $I$. }
    \label{fig:L1}
\end{figure*}

\begin{figure*}
    \centering
    \includegraphics{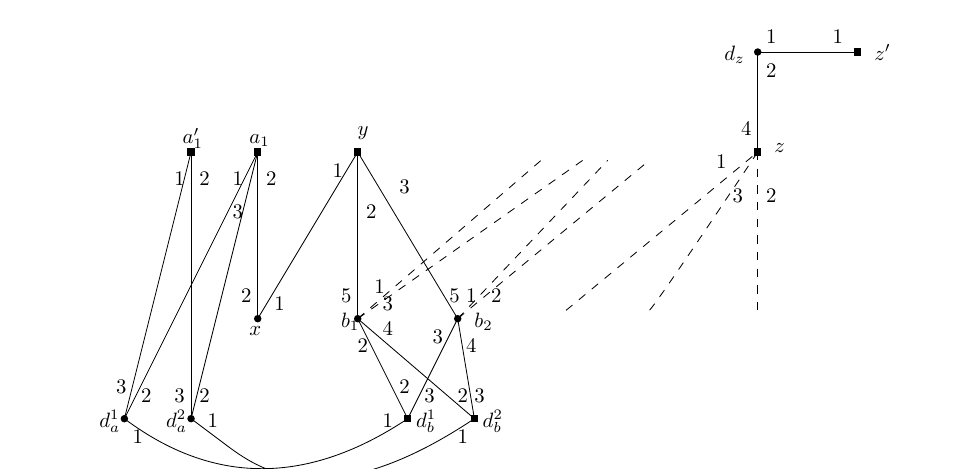}
    \caption{The instance $I'$. The numbers illustrate the preferences in the preference profile $L_2$ created from $I$.}
    \label{fig:L2}
\end{figure*}

We claim that there is a certainly popular matching $M'$ in $I'$ if and only if there is a certainly popular matching $M''$ in $I''$ if and only if there is a popular matching $M$ in $I$ that covers $z$ and includes the edge $(x,y)$.

\begin{lemma}
    If there is a matching $M'$ that is popular in both $L_1$ and $L_2$ in $I'$, then there is a matching $M$ that is popular in $I$, covers $z$ and contains $(x,y)$.
\end{lemma}
\begin{proof}
    Let $M'$ be a matching that is popular in both $L_1$ and $L_2$. Set $M$ to be the matching we obtain by restricting $M'$ to $G$, that is $M=\{ (u,v)\in M' \mid u,v\in V(G)\}$.

    \begin{claim}
        $z$ is covered in $M$.
    \end{claim}
\begin{claimproof}
    Suppose for the contrary that $z$ is not covered in $M$. 
Then, if $(z,d_z)\in M'$, then $z'$ is unmatched in $M'$, and therefore $M'\setminus \{ (z,d_z) \} \cup \{ (z',d_z)\} $ dominates $M'$ in $L_2$, because $d_z$ and $z'$ votes with $-1$, $z$ votes with $+1$ and everyone else votes with 0. 

    Otherwise, if $(z,d_z)\notin M'$, then $z$ in unmatched in $M'$, so $(z',d_z)\in M'$ and $M'\setminus \{ (z',d_z) \} \cup \{ (z,d_z)\} $ dominates $M'$ in $L_1$, as $z$ and $d_z$ vote with $-1$ and only $z'$ votes with $+1$.
\end{claimproof}
\medskip 

Note that this also shows that $(z',d_z)\in M'$.
    
Next we show that $(x,y)$ must be included in $M$. 

\begin{claim}
    It holds that $(x,y)\in M$. 
\end{claim}
\begin{claimproof}
    Suppose for the contrary that $(x,y)\notin M$. Then, $(x,y)\notin M'$. As $x,y$ consider each other best by the assumptions of \popf, we get that both of them must be matched in $M'$, so $(x,a_1)\in M'$ and $(y,b_i)\in M'$ for some $i\in [l]$. 
    
    If $(d_a,d_b)\in M'$, then $M'\setminus \{ (x,a_1),(y,b_i),(d_a,d_b)\} \cup \{ (x,y),(a_1,d_a),(b_i,d_b)\}$ dominates $M'$ in $L_2$, as $x,y,a_1,b_i$ all vote with $-1$, and only $d_a,d_b$ vote with $+1$, contradiction.

    Otherwise, as $d_a,d_b$ consider each other best in both $L_1$ and $L_2$, each must be matched in $M'$, so $(a_1',d_a),(b_j,d_b)\in M'$. But then,  $M'\setminus \{ (x,a_1),(y,b_i),$ $(a_1',$ $d_a),(b_j,d_b)\} \cup \{ (x,y),(a_1,d_a),(b_i,d_b)\}$ dominates $M'$ in $L_2$, because $x,y,a_1,b_i,d_a$ vote with $-1$ and only $a_1',d_b,b_j$ can vote with $+1$, contradiction. 
\end{claimproof}

\begin{claim}
    $M$ is popular in $I$.
\end{claim}
\begin{claimproof}
    Suppose for the contrary that there is a matching $N$ in $G$, such that $N$ dominates $M$ in $I$. Then, create a matching $N'$ from $N$, by adding the edges $\{ (d_a,d_b),(z',d_z)\}$.

    It is clear that when comparing $N'$ to $M'$, agents $z',d_z$ vote with 0, and $d_a,d_b$ vote with $\le 0$, as they get their best partner in $N'$ in both $L_1$ and $L_2$. Agent $a_1'$ may vote with $+1$, but if and only if he was matched to $d_a$ in $M'$ (and he is unmatched in $N'$), in which case the $-1$ vote of $d_a$ cancels it to sum up to $0$. Hence, the sum of votes of the agents $a_1',z,d_z,d_a,d_b$ are at most $0$ in $L_1$.

    Suppose that for each agent $v\in V(G)$, $\vote_v^{L_1} (M',N')\le \vote_v^{I}(M,N)$. Then, by our above observation we get that $N'$ dominates $M'$ too in $L_1$, contradiction. 

    Therefore, suppose that there is an agent $v\in V(G)$, such that $\vote_v^{L_1} (M',N')> \vote_v^{I}(M,N)$. Such an agent must have a better situation in $M'$ than in $M$, since each $v\in V(G)$ has the same partner in $N$ and $N'$. Therefore, $v$ can only be $a_1,b_i$ ($i\in [l])$ or $z$. We show that if $\sum_{v\in V(G)}\vote_v^{L_1} (M',N')-\sum_{v\in V(G)}\vote_v^{I}(M,N)=\eta >0$, then $\vote_{a_1'}^{L_1}(M,N)+\vote_{d_a}^{L_1}(M,N)+\vote_{d_b}^{L_1}(M,N)+\vote_{z'}^{L_1}(M,N)+\vote_{d_z}^{L_1}(M,N)\le -\eta$.

    If $v=z$, we know that $z$ had a partner from $V(G)$ in $M'$. Hence, he has the same situation in $M$ and $M'$, contradiction. 

    Suppose that $v=b_i$, for some $i\in [l]$. As $b_i$ has a better situation in $M'$ than in $M$ in $L_1$, we get that $b_i$ must be unmatched in $M$, but matched to $d_b$ in $M'$. However, as $d_b$ is worse than any $v\in V(G)$ neighbor in $L_1$, we obtain that if $b$ voted with $-1$ in the comparison $M$ versus $N$ in $I$, then he got a partner $v\in V(G)$ in $N$, so also votes with $-1$ in $M'$ versus $N'$ in $L_1$, contradiction. If $b_i$ was unmatched in $N$, then $b_i$ voted with $0$ in $M$ versus $N$, but he votes with $+1$ in $M'$ versus $N'$. However, then we have an additional $-1$ vote of $d_b$ that cancels this increase of $1$ to $\vote_{b_i}^I(M,N)$, so the sum of votes remain less than $0$.
    
    Finally, suppose that $v=a_1$. Then, $a_1$ was unmatched in $M$ (as $a_1$ had one neighbor $x$ and $(x,y)\in M$). If $(a_1,d_a)\notin M'$, then $a_1$'s vote does not differ in the two votings. Hence, suppose that $(a_1,d_a)\in M'$. Then, $(b_i,d_b)\in M'$ for some $i\in [l]$.

    Suppose first that $\vote_{a_1}^I(M,N)=-1$. This is only possible, if $a_1$ got matched to $x$ in $N$. Also, now $a_1$ votes with $+1$ in $M'$ versus $N'$ in $L_1$. First assume that $b_i$ was unmatched in $N$. Then, $\vote_{b_i}^I(M,N)=0$.
    In this case, we modify $N'$ to $N''$ as follows: delete $(d_a,d_b)$ and add $(a_1',d_a)$ and $(b_i,d_b)$. The restriction of this matching to $G$ is still $N$ and only the votes of $a_1,b_i,d_a,d_b$ may be different. Then, $a_1$ still votes with $+1$ in $M'$ versus $N'$ in $L_1$, but $d_a$ and $a_1'$ vote with $-1$, which together cancel out the increase of at most $2$ in $a_1$-s vote from $\vote_{a_1}^I(M,N)$. Also, $b_i$ and $d_b$ vote with 0, so $b_i$'s vote does not change. Hence, as $z'$ and $d_z$ also vote with 0, we obtain that this matching $N''$ dominates $M'$ in the profile $L_1$, as $\sum_{v\in \{ a_1',z,d_z,d_a,d_b\} } \vote_v^{L_1}(M',N'')\le -2$ and $\sum_{v\in V(G) } \vote_v^{L_1}(M',N'')<0 + 2=2$ contradiction.
    
     Next suppose that $b_i$ was matched in $N$. Then, $b_i$ is also matched in $N'$ to the same partner and has the same vote $-1$, as any $v\in V(b_i)$
   is better for $b_i$ than $d_b$ in $L_1$. Also, $a_1'$ was not matched to $d_a$ in $M'$, so we do not have to use $d_a$'s $-1$ vote to cancel $a_1'$'s vote, as it cannot be $+1$. Hence, we can use the $-1$ votes of $d_a$ and $d_b$ in $M'$ versus $N'$ to cancel the possible increase of $2$ in the vote of $a_1$, as the $-1$ vote of $d_b$ also does not need to be used to cancel the vote of $b_i$ now. 

   Finally, suppose that $vote_{a_1}^I(M,N)=0$, so $a_1$ was unmatched in both $M$ and $N$, but $\vote_{a_1}^{L_1}=+1$. Then, $(a_1,d_a)\in M'$, $(a_1',d_a)\notin M'$ and $(d_a,d_b)\notin M'$, so we can use the $-1$ vote of $d_a$ to cancel the increase of $1$ in the vote of $a_1$, as it is not used to cancel any other $+1$ votes (since in this case we have $\vote_{a_1'}^{L_1}(M,N)=0$).

   Therefore, we conclude that either $N'$ or $N''$ dominates $M'$ in the profile $L_1$, contradicting that $M'$ is certainly popular.
\end{claimproof}    
\end{proof}

Now we show that a popular matching of $I$ that satisfies the constraints gives a popular matching for any preference profile in $I''$.

\begin{lemma}
    If there is a popular matching $M$ in $I$ that covers $z$ and contains $(x,y)$, then there is a matching $M'$ that is popular with respect to any preference profile in $I''$.
\end{lemma}
\begin{proof}

Let $M$ be a matching that is popular in $I$, covers $z$ and contains $(x,y)$. Extend it to a matching $M'$, by adding the edges $\{ (d_a,d_b), (z',d_z) \}$. 

We claim that $M'$ is popular with respect to any possible preference profile $L$ in the instance $I''$. 

 Suppose for the contrary that there is a possible preference profile $L$, such that there exists a matching $N'$ that dominates $M'$ and let $N$ be the restriction of $N'$ to $G$. 

Let $\eta = \sum_{v\in V(G)}\vote_v^I(M,N)\ge 0$, as $M$ is popular.
Note that $d_a$ and $d_b$ have their only possible best partner matched to them in $M'$, so they vote with $+1$, whenever they obtain a different partner, otherwise, they vote with $0$. For agent $d_z$, he can only vote with $-1$ if $(d_z,z)\in N'$, but in this case, $z'$ casts a $+1$ vote that cancels it. For agent $a_1'$, if he casts a vote of $-1$, then he must get a $d_a$ as a partner in $N'$, who then votes with $+1$ that cancels it. 

Therefore, if $\vote_v^L(M',N')\ge \vote_v^I(M,N)$ for each agent $v\in V(G)$, then $\Delta^L (M',N')\ge 0$, contradicting that $N'$ dominates $M'$. 

Suppose there is an agent $v\in V(G)$ that casts a smaller vote, so $\vote_v^L(M',N')< \vote_v^I(M,N)$. As each such agent has the same situation in both $M$ and $M'$, this is only possible if $v$ obtains a better partner in $N'$ than in $N$. 

Suppose this agent is $z$. Then, $z$ must be unmatched in $N$ and be matched to $d_z$ in $N'$, otherwise either he would have the same partners in both matchings or would not prefer $N'$ to $N$. However, we know that $z$ was matched in $M$, so he was matched to a better partner in $M'$ than $d_z$, so he casts the same $+1$ vote in both scenarios. 

Suppose this agent is $a_1$. This means that $a_1$ was unmatched in $N$, but is matched to some $d_a$ in $N'$. Then, $vote_{a_1}^I(M,N)=0$, as $a_1$ must be unmatched in $M$ too, since $(x,y)\in M$. Furthermore, if $\vote_{a_1}^{L_1}(M',N')=-1$ in the vote $M'$ versus $N'$, then this decrease of $1$ in $a_1$'s vote can be canceled with the $+1$ vote of his partner $d_a$'s in $N'$ (using that in this case $\vote_{a_1'}^I(M,N)=\vote_{a_1'}^{L_1}(M',N')=0$).

Suppose lastly that this agent is $b_i$, for some $i\in [l]$. Again, this means that $b_i$ is unmatched in $N$ but is matched in $N'$, otherwise $b_i$ has the same vote in both situations. However, this implies that he must be matched to $d_b$ in $N'$, and therefore his possible $-1$ vote can be canceled with the $+1$ vote of his partner $d_b$.

In summary, we concluded that the sum of votes must remain nonnegative in the vote $M'$ versus $N'$, contradicting the fact that $N'$ dominated $M'$.
\end{proof}
Finally, we note that it is trivial that if there is a matching that is certainly popular in $I''$, then it is also popular in both preference profiles $L_1$ and $L_2$ in $I'$. Therefore, the theorem follows from the previous two Lemmata. 
\end{proof}



Next we proceed to the hardness of \poprob.

\begin{theorem}
\label{thm:robustpop}
    \poprob\ is NP-complete, even for a fixed constant $k=1$. 
\end{theorem}
\begin{proof}
It is in NP by Theorem \ref{thm:2sideverif}.
    Let $k=1$.
    We reduce from the NP-complete problem \textsc{super-pm}. Let $I$ be an instance of \textsc{super-pm}, where each agent either has a strict preference list, or has only two neighbors that he puts in a single tie. By the reduction of \cite{csaji2023weaklypopular}, this is also NP-complete. 

    Let $G=(A,B,E)$ be the graph of \textsc{super-pm}, with $A=\{ a_1,\dots, a_n\}$ and $B=\{ b_1,\dots,b_m\}$. Suppose without loss of generality that $n\ge m$. Let $A'\subset A$ and $B'\subset B$ be the sets of agents of $A$ and $B$ with strict preference lists.
    
    We add four sets of agents $U=\{ u_1,\dots,u_n\} $, $U'=\{ u_1',\dots,u_n'\}$, $W=\{ w_1,\dots,w_m\}$ and $W'=\{ w_1',\dots, w_m'\}$. 

    Now, we extend the original preferences of each agent $a_i\in A$ with a strict preference list by inserting $u_1$ between his first and second choice, $u_2$ between his second and third choice, etc, and put the remaining $u_i$ agents to the bottom of the preference list in the order of their indices. For example, if the preference list of $a_i$ was $b_3\succ b_2\succ b_5$, then now it becomes $b_3\succ u_1\succ b_2\succ u_2\succ b_5 \succ u_3\succ \dots \succ u_n$. We extend the preferences of the agents in $B$ similarly, by inserting the ones from $W$. For the agents in $I$ with a tie, we just break the tie and not add any more agents to their preference lists.

    Finally, we create the preferences of the added new agents. For $u_i$, we let his preference list be $u_i'\succ u_{i+1}'\succ u_{i+2}'\succ \dots \succ u_{i-1}' \succ [A']$, where $i+1$ is taken as $i \; (mod \; n) \; +1$, and $[A']$ denotes an arbitrary strict ordering of the agents in $A'$. For $w_j$, we create it as $w_j'\succ w_{j+1}'\succ \dots \succ w_{j-1}'\succ [B']$, where $j+1$ is taken as $j \; (mod \; m)\; +1$. Finally, for $u_i'$, we have $u_i\succ u_{i+1}\succ u_{i+2}\succ \dots \succ u_{i-1}$ and for $w_j'$ we have $w_j\succ w_{j+1}\succ \dots \succ w_{j-1}$.

    \begin{claim}
        Suppose that there is a $1$-robust popular matching in $I'$. Then, there is a super popular matching in $I$.
    \end{claim}
    \begin{claimproof}
        Let $M'$ be a 1-robust matching. We claim that $\{ (u_i,u_i'), (w_j,w_j')\mid i\in [n],j\in [m]\} \subset M'$. Suppose that $(u_i,u_i')\notin M'$. Then, as they are mutual best, $x=M'(u_i)$ and $y=M'(u_i')$ exist, and hence none of them is with their best partner in $M'$. Then, $u_i,u_i',x$ and $y$ can all get matched to their favorite, so strictly better partners, and only the original partners of $x$ and $y$ get left alone in doing so, so $M'$ is not even popular, contradiction. Similar reasoning works for $(w_j,w_j')\notin M'$.

        Hence, each edge of $M=M'\setminus \{ (u_i,u_i'), (w_j,w_j')\mid i\in [n],j\in [m]\}$ must be an edge of $G$, so $M$ is a matching in $G$ and every $u\in A\cup B$ has the same partner in $M$ and $M'$. We claim that $M$ is super-popular. Suppose that some matching $N$ dominates it. Let $N'=N\cup \{ (u_i,u_i'), (w_j,w_j')\mid i\in [n],j\in [m]\}$. 

        We claim that each agent has the same vote when comparing $N$ to $M$; and $N'$ to $M'$.
For the agents in $U\cup U'\cup W \cup W'$, this is trivial, as their vote is always 0. For the agents in $A'\cup B'$, it is also clear, because between any two of their neighbors in $G$, there is an agent $u_i\in U$ or $w_i\in W$, so if they get a better or worse partner in $N$ than in $M$ in $I$, then in $I'$ this partner is at least $2$-better or $2$-worse, so definitely better or definitely worse respectively. For the agents who had a single tie, they had only two neighbors in $I$ and also only two neighbors in $I'$. If such an agent is matched in $M$ but is not matched in $N$ or is unmatched in $M$ but is matched in $N$ or he has the same partner (possibly no one) in both, then his votes in $I$ and $I'$ are the same. Otherwise, if he switches between his two partners, then he also votes with $-1$ in both $I$ by the definition of super-popularity and in $I'$ by the definition of 1-robust popularity. 

Hence, this implies that $N'$ dominates $M'$ in $I'$, contradiction.
    \end{claimproof}

    \begin{claim}
        If there is a super-popular matching in $I$, then there is a 1-robust matching in $I'$. 
    \end{claim}
    \begin{claimproof}
Let $M$ be a super-popular matching in $I'$. Create $M'=M\cup \{ (u_i,u_i'), (w_j,w_j')\mid i\in [n],j\in [m]\}$.

Suppose that $M'$ is not 1-robust popular and let $N'$ be a matching that dominates it. We claim that we may assume that $\{ (u_i,u_i'), (w_j,w_j')\mid i\in [n],j\in [m]\} \subset N'$. Suppose that each edge of $N'$ is either an edge of $G$, or an edge of type $(u_i,u_j')$ or $(w_i,w_j')$. In this case, the sum of votes on each edge $(u_i,u_j')$ and $(w_i,w_j')$ is at most 0, because if $i\ne j$, then at least one of them has a partner that is more than $1$-worse. So in this case, if we take the restriction of $N'$ to $G$ and add $\{ (u_i,u_i'), (w_j,w_j')\mid i\in [n],j\in [m]\}$, then it still dominates $M'$, but now this matching satisfies our assumption. 

So let us suppose that there is an edge $(u_i,a_j)\in N'$. Let the number of such edges be $\eta$. This means that there must be at least $\eta$ unmatched agents from $U'$ in $N'$. Also, in an $(u_i,a_j)$ edge, $u_i$ votes with $+1$ in $I'$. Hence, if we drop each edge adjacent to any $u_i\in U$ or $u_i'\in U'$ from $N'$ and add $\{ (u_i,u_i')\mid i\in [n]\}$, then the sum of the votes of the $\eta$ $a_i$ agents may increase by at most $2\eta$, but the sum of the votes in $U\cup U'$ decreases by at least $2\eta$ too. Hence, this matching also dominates $M'$. Similar arguments with assuming $(w_i,b_j)\in N'$ show that indeed we may assume that $\{ (u_i,u_i'), (w_j,w_j')\mid i\in [n],j\in [m]\} \subset N'$.

Therefore, let us take $N=N'\setminus \{ (u_i,u_i'), (w_j,w_j')\mid i\in [n],j\in [m]\}$. Then, each agent in $A\cup B$ has the same partners in $M$ and $M'$ and also in $N$ and $N'$. However, with the same reasoning as in the previous claim, we get that now each agent in $A\cup B$ has the same vote in $M$ versus $N$ in $I$ and in $M'$ versus $N'$ in $I'$. As each agent in $U\cup U'\cup W\cup W'$ votes with $0$, this implies that $N$ dominates $M$ in $I$, contradiction.

    \end{claimproof}

The theorem follows from the two claims. For $k>1$, a similar reduction works with more dummy agents, where we insert at least $k$ of them between any two entries in a preference list of an agent from $A'\cup B'$. 
\end{proof}

We conclude this section with the hardness of \popMLsum\ and \popunsum.

\begin{theorem}
\label{thm:popMLsum}
    \popMLsum\ and \popunsum\ are NP-complete even if only one side has uncertain preferences and a fixed constant $k=2$. \domMLsum\ and \domunsum\ are also NP-complete, even in the same setting. 
\end{theorem}
\begin{proof}
It is in NP by Theorem \ref{thm:2sideverif}.

    We reduce from the NP-hard \popties\ \cite{cseh2017popular}, where we are given an instance $I$ of the popular matching problem with ties in the preference lists. We may assume that ties are only occurring on side $A$ and each agent has either a strict preference list or puts all his neighbors in a single tie by \cite{cseh2017popular}. 
    
    The reduction itself is quite simple. Let $I$ be an instance of \popties. 
    Create an instance $I'$ of \popMLsum\ as follows. 
    For each agent $u\in A\cup B$ with strict preference list $\succ_u$ we keep his preference list in both $L_1$ in $L_2$. For each $u\in A$ with a tied preference list, we break the ties in an arbitrary way to obtain his preference list in $L_1$, then break the ties in the opposite way and obtain his preference list in $L_2$. 

    We claim that for any matchings $M$, $N$, and any agent $u$ it holds that $2\cdot \vote_u^I(M,N)=\vote_u^{L_1}(M,N)+\vote_u^{L_1}(M,N)$. 

    For an agent $u\in A\cup B$ with a strict preference list, we have that $\vote_u^I(M,N)=\vote_u^{L_1}(M,N)=\vote_u^{L_2}(M,N)$. For an agent $u\in A$ with a single tie, if $u$ is not matched in $M$, but matched in $N$, we get that $\vote_u^I(M,N)=\vote_u^{L_1}(M,N)=\vote_u^{L_2}(M,N)=-1$. Similarly, if $u$ is matched in $M$, but not matched in $N$, we get that $\vote_u^I(M,N)=\vote_u^{L_1}(M,N)=\vote_u^{L_2}(M,N)=+1$. Otherwise, we have that $\vote_u^I(M,N)=0$ and that $\vote_u^{L_1}(M,N)=-\vote_u^{L_2}(M,N)$, because in this case, whenever $\vote_u^{L_1}(M,N)\ne 0$, we get that $u$ has different partners in $N$ and $M$, and by the construction of $u$'s preferences in $L_1$ and $L_2$, he votes with $+1$ in one case and $-1$ in the other. 

    Hence, we obtained that $2\cdot vote_u^I(M,N)=vote_u^{L_1}(M,N)+vote_u^{L_2}(M,N)$ for each $u\in A\cup B$. From this we obtain that $M$ is popular in $I$, if and only if it is sum-popular in $I'$.

    To show the statement of the theorem for \domMLsum\ and \domunsum\ we observe that the reduction of \cite{cseh2017popular} satisfies that if there is a popular matching, then it must be a maximum size matching. Hence, for the instances created in the reduction a matching is dominant if and only if it is popular, as there is no larger matching. Hence, the statement follows.
\end{proof}

\section{Conclusions}\label{sec:concl}

We studied popular and dominant matchings under four representations of preference variation: independent uncertainty, multiple layers, bounded swap perturbations, and aggregation. The results expose a sharp market-structure boundary. Popular-matching existence is NP-hard throughout the two-sided models, whereas all four one-sided variants are polynomial-time solvable; in the robust case this follows from a polynomial reduction to one-sided uncertainty. For aggregation, we extended the popular-common-base level algorithm of Kavitha et al.~\cite{yu2023arborescence} to bounded integral skew-symmetric comparison margins. The aggregate house-allocation model reduces to this extension, which handles any explicitly listed number of profiles despite the Condorcet cycles that arise from three profiles onward. For dominant matchings, uncertainty and robustness remain tractable through their connection with stable matchings, while multilayer and aggregated preferences lead to NP-hardness.

The weak-preference extension shows that ties do not destroy the one-sided tractability boundary. A layer-counting weight function enforces maximum first-choice matchings simultaneously, and a reachability construction identifies assignments that can serve as pseudo-second choices in every realization. This yields polynomial algorithms for uncertain and multilayer one-sided markets with ties without enumerating the exponentially many profiles encoded by independent uncertainty sets; the robust result follows through the same reduction to uncertainty. The integral-margin common-base algorithm also permits ties, since it requires only bounded integral skew-symmetric comparison margins.

For future work, one might consider optimization versions, such as finding matchings that are popular with highest probability, or consider other models of uncertainty.

\bibliographystyle{amsplain}
\bibliography{references}

\end{document}